\documentclass[12pt,draftcls,onecolumn]{IEEEtran}
\usepackage{amsmath,amssymb,amsthm}    
\usepackage[dvips]{graphicx}   
\usepackage{verbatim}   
\usepackage{color}      
\usepackage{subfigure}  
\usepackage{epsfig}
\usepackage{float}
\usepackage{algorithm}
\usepackage{algorithmic}
\usepackage{setspace}
\usepackage{subfigure}
\usepackage{cite}
\newtheorem{lemma}{\bf Lemma}
\newtheorem{theorem}{\bf Theorem}
\newtheorem{corollary}{\bf Corollary}
\newtheorem{proposition}{\bf Proposition}
\newtheorem{definition}{\bf Definition}

\makeatletter
\newcommand*{\rom}[1]{\expandafter\@slowromancap\romannumeral #1@}
\makeatother
\begin{document}
\title{\textbf{Collaborative Compressive Detection with  Physical Layer Secrecy Constraints }}
\author{Bhavya~Kailkhura,~\IEEEmembership{Student Member,~IEEE}, Thakshila~Wimalajeewa,~\IEEEmembership{Member,~IEEE}, Pramod~K.~Varshney,~\IEEEmembership{Fellow,~IEEE}
\thanks{This work was supported in part by the National Science Foundation (NSF) under Grant No. 1307775.}
\thanks{Some related preliminary work was presented at the Forty Eighth Asilomar
Conf. on Signals, Systems and Computers, Nov 2014.}
\thanks{ 
Authors are with Department of EECS, Syracuse University, Syracuse, NY 13244. (email: bkailkhu@syr.edu; twwewelw@syr.edu; varshney@syr.edu).}
}
\maketitle
\begin{abstract}
This paper considers the problem of detecting a high dimensional signal (not necessarily sparse) based on compressed measurements with physical layer secrecy guarantees. First, we propose a collaborative compressive detection (CCD) framework to compensate for the performance loss due to compression with a single sensor. We characterize the trade-off between dimensionality reduction achieved by a universal compressive sensing (CS) based measurement scheme and the achievable performance of CCD analytically.
Next, we consider a scenario where the network operates in the presence of an eavesdropper who wants to discover the state of the nature being monitored by the system. 
To keep the data secret from the eavesdropper, we propose to use cooperating trustworthy nodes that assist the fusion center (FC) by injecting artificial noise to deceive the
eavesdropper. We seek the answers to the questions: Does CS help improve the security performance in such a framework? What are the optimal values of parameters which maximize the CS based collaborative detection performance at the FC while ensuring perfect secrecy at the eavesdropper? 
\end{abstract}


\section{Introduction}
Compressive sensing (CS) is a new paradigm which enables the reconstruction of compressible or sparse signals using far
fewer samples than required by the Nyquist criterion~\cite{donoho,candes}. In this framework, a small collection of linear random projections
of a sparse signal contains sufficient information for signal recovery. To reconstruct the original signal from its compressed measurements, several algorithms have been proposed in the literature~\cite{cssurvey}.

While CS mostly deals with complete signal reconstruction, there are several signal processing applications
where complete signal recovery is not necessary. Instead
we might be only interested in solving inference problems such as detection, classification or estimation of certain parameters. 
To solve an inference problem where some prior information about the signal is available, a customized measurement scheme could be implemented such that the optimal inference performance is achieved for the particular signal. As an example, for a signal detection problem where the signal of interest is known, the optimal design is the matched filter which is dependent on the signal itself. However, it is possible that the signal that we wish to infer about may evolve over time. Thus, we are often interested in universal or agnostic design. A few attempts have been made in this direction to address the problems of inference in Compressive Signal Processing (CSP) literature in recent research~\cite{dave,haupt,thak}. 
CSP techniques are universal and agnostic to the signal structure and provide deterministic guarantees for a wide variety of signal classes. 

The authors in~\cite{dave,bcd,durate} considered the deterministic signal detection problem in
the compressed measurement domain where the performance limits of detection with compressed measurements were investigated. 
For signals that are not necessarily sparse, it was shown that a certain performance loss will be incurred due to compression when compared to the optimal test that acquires original measurements using the traditional measurement scheme.  
For stochastic signals, the compressive detection problem
(i.e., detecting stochastic signals in the compressed measurement domain) was considered in~\cite{thak,gen}. Both works focused only on compressive detection of `zero-mean' stochastic signals based on observations corrupted by additive noise. Closed form expressions were derived for performance
limits and performance loss due to compression was characterized analytically. 
A signal classification problem based on compressed measurements was considered in~\cite{smash} where the authors developed a manifold based model for compressive classification. 
The authors in~\cite{haput1,haput2} studied the performance of compressive sampling in detection and classification setup and introduced the generalized restricted isometry property that states that the angle between two vectors is preserved under random projections.
Sparse event detection by sensor networks under a CS framework was considered in~\cite{hansparse}. The problem of detection of spectral targets based on noisy incoherent projections was addressed in~\cite{hyper}.  
Schemes for the design and optimization of projection matrices 
for signal detection with compressed measurements have been proposed in~\cite{vila,bai,shi,pso}. 

As mentioned earlier, CSP techniques are universal and agnostic to the signal structure and, therefore, are attractive in many practical applications. 
Despite its attractiveness to solve high dimensional inference problems, CSP suffers from a few major drawbacks which limit its applicability in practice. A CS based measurement scheme incurs a certain performance loss due to compression when compared to the traditional measurement scheme while detecting non sparse signals. This can be seen as the price one pays for universality in terms of inference performance. 
In this paper, we propose a collaborative compressive detection (CCD) framework to compensate for the performance loss due to compression. The CCD framework comprises of a group of spatially distributed nodes which acquire vector observations regarding the phenomenon of interest. Nodes send a compressed summary of their observations to the Fusion Center (FC) where a global decision is made. 
In this setup, we characterize the trade-off between  dimensionality reduction in a universal CS based measurement scheme and the achievable performance. In our preliminary work~\cite{bhavyaphy}, we analyzed the problem
only for the deterministic signal case. In the current work, we significantly extend our previous work and investigate the problem for two different cases: $1)$ when the signal of interest is deterministic and $2)$ when the signal of interest is random. It is worthwhile to point out that, in contrast to~\cite{thak,gen} where compressive detection of `zero-mean' stochastic signals was considered, we study a more general problem with `non zero-mean' stochastic signals. Note that, some of these existing results can be seen as a special case of analytical results derived in this paper. 
For both the cases, we show that for a fixed signal to noise ratio (SNR), if the number of collaborating nodes is greater than $(1/c)$, where $0\leq c\leq 1$ is the compression ratio, the loss due to compression can be recovered.

In a CCD framework, the FC receives compressive observation vectors from the nodes and makes the global decision about the presence of the signal vector. The transmissions by the nodes, however, may be observed by an eavesdropper. The secrecy of a detection system against eavesdropping attacks is of utmost importance~\cite{physecdet}.
In a fundamental sense, there are two motives for any eavesdropper (Eve), namely \emph{selfishness} and \emph{maliciousness}, to compromise the secrecy of a given inference network. For instance, some of the nodes within a cognitive radio network (CRN) may selfishly take advantage of the FC's inferences and may compete against the CRN in using the primary user's channels without paying any participation costs to the network moderator. In another example, if the radar decisions are leaked to a malicious aircraft, the adversary aircraft can maliciously adapt its strategy against a given distributed radar network accordingly so as to remain invisible to the radar and in clandestine pursuit of its mission. Therefore, in the recent past, there has been a lot of interest in the research community in addressing eavesdropping attacks on inference networks. Recently, a few attempts have been made to address the problem of eavesdropping threats on distributed detection network. However, a similar study in a CSP framework is missing from the literature.

Next, we investigate the CCD problem when the network operates in the presence of an eavesdropper who wants to discover the state of the nature being monitored by the system. 
While security issues with CS based measurement schemes have been considered in~\cite{cssec1,cssec2,cssec3}, our work is considerably different. 
In contrast to~\cite{cssec1,cssec2,cssec3}, where performance limits of secrecy of CS based measurement schemes were analyzed (under different assumptions), we look at the problem from a practical perspective. We pursue a more active approach where the problem of optimal system design with secrecy guarantees is studied in an optimization setup.
More specifically, we propose to use cooperating trustworthy nodes that assist the FC by injecting artificial noise to deceive the eavesdroppers to improve the security performance of the system. The addition of artificial noise to node transmissions is a data falsification scheme that is employed to confuse the eavesdropper.
We consider the problem of determining optimal system parameters which maximize the detection performance at the FC, while ensuring perfect secrecy at the eavesdropper (information of the eavesdropper is exactly zero). In the process of determining optimal system parameters, we seek the answer to the question: Does compression help in improving the security performance of the system? 
At first glance, it seems intuitive that compression should
always improve the security performance. However, we show that this argument is not necessarily
true. In fact, security performance of the system is independent of the compression ratio in the perfect secrecy regime.

\subsection{Main Contributions}

Our work presented in this paper is motivated by a number of factors. First, CS based measurement schemes incur
a certain performance loss due to compression when compared
to traditional measurement schemes while detecting non sparse
signals, and, therefore, techniques to mitigate this loss are desirable.
Next, CSP has been proposed relatively recently and, therefore, security issues for such a framework have been left un-addressed so far.
Also, most of the works on CSP mainly focus on deriving theoretical performance bounds under different contexts. Despite its theoretical importance, practical implications of these bounds for system design have not been investigated in literature. 
In this paper, we take some first steps in addressing these issues for solving high dimensional signal detection problems using only compressed measurements taking security aspects into consideration. The main contributions of this paper are summarized as
follows.
\begin{itemize}
\item We propose a collaborative compressive detection (CCD) framework to compensate for the performance loss due to compression.
\item We characterize the trade-off between
dimensionality reduction in a universal compressive sensing
based measurement scheme and the achievable performance
of CCD analytically.
\item When the network operates in the presence of an eavesdropper, we employ artificial noise injection techniques to improve secrecy performance. Theoretical performance bounds for the scheme are also derived.
\item We consider the problem of determining optimal system parameters which maximize the detection performance at the FC, while ensuring perfect secrecy at the eavesdropper.
\end{itemize}

The rest of the paper is organized as follows. Section~\ref{sec2} presents the observation model and the problem formulation. In Section~\ref{sec3}, performance of collaborative compression detection is analyzed for both deterministic and random signal cases. In Section~\ref{sec4}, we investigate the problem where the network operates in the presence of an eavesdropper and propose artificial noise injection techniques to improve secrecy performance. In Section~\ref{sec5}, we study the problem of determining optimal system parameters which maximize the detection performance at the
FC, while ensuring perfect secrecy at the eavesdropper.
Concluding remarks and possible future directions are given in Section~\ref{sec6}.

\section{Collaborative Compressive Detection }
\label{sec2}
\begin{figure}[!t]
\centering
    \includegraphics[width=3in]{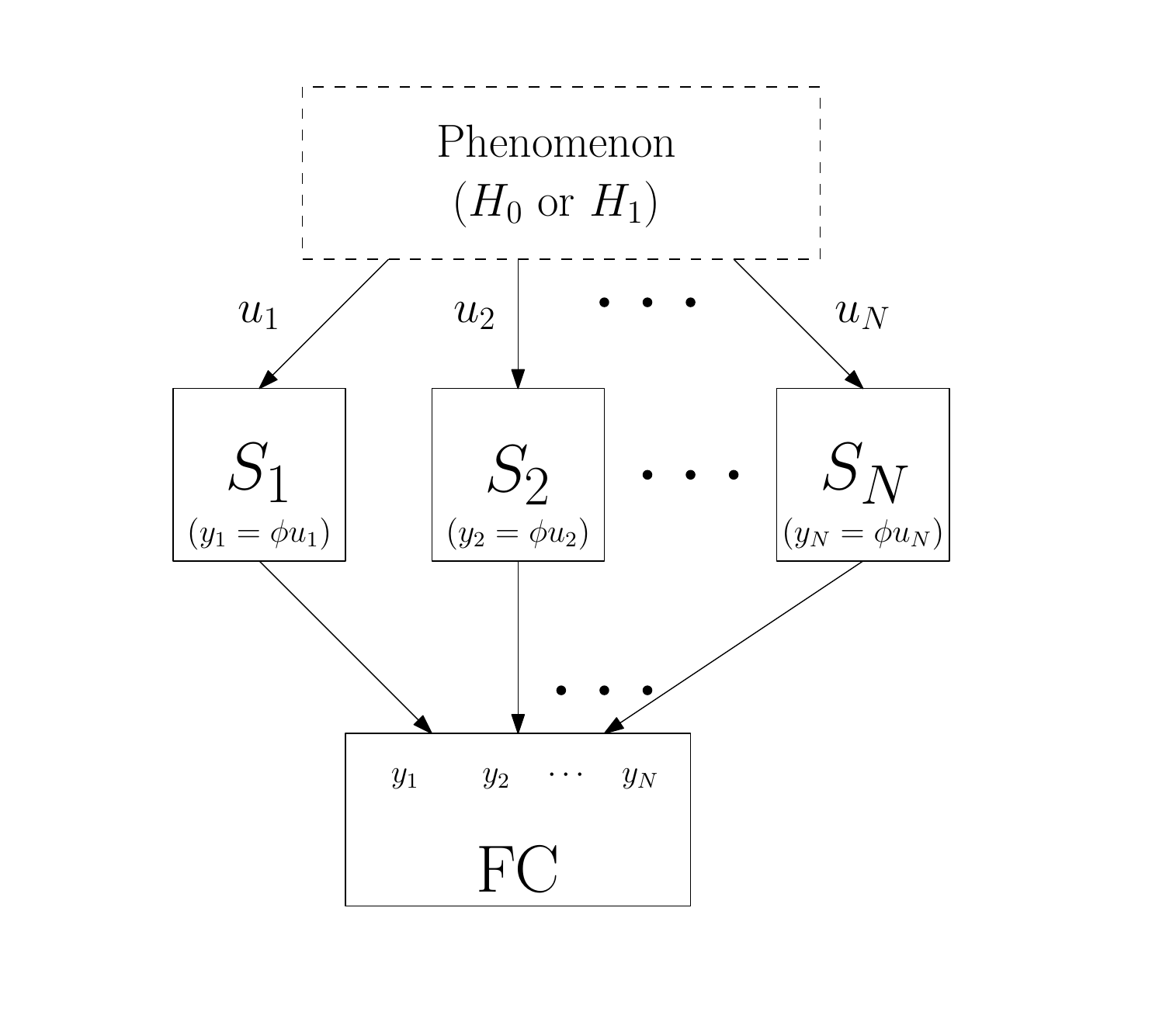}
\vspace{-0.2in}
    \caption{Collaborative Compressive Detection Network}
    \label{Fig: model}
\end{figure}

\subsection{Observation Model}
Consider two hypotheses  $H_{0}$ (signal is absent) and $H_{1}$ (signal is present).
Also, consider a parallel network, comprised of a central entity (known as the FC)
and a set of $N$ nodes, which
faces the task of determining which of the two hypotheses is true (see Figure~\ref{Fig: model}).
Prior probabilities of the two hypotheses $H_{0}$ and $H_{1}$
are denoted by $P_{0}$ and $P_{1}$, respectively. 
The nodes observe the phenomenon (high dimensional signal), carry out local compression (low dimensional projection), 
and then send their local summary statistic to the FC. The FC makes a final decision after processing the locally compressed observations.

For the $i$th node observed signal, $u_i$ can be modeled as
\begin{center}
\vspace{-0.2in}
\begin{eqnarray*}
&& H_0\; :\quad u_i=v_i\\
&& H_1\; :\quad u_i=s+v_i
\end{eqnarray*}
\end{center}
where $u_i$ is the $P \times 1$ observation vector, $s$ is  either deterministic or random Gaussian signal vector (not necessarily sparse) to be detected. Specifically let $s \sim \mathcal{N} (\mu,\alpha^{-1} I_P)$ and additive noise $v_i \sim \mathcal{N} (0,\beta^{-1} I_P)$
where $x\sim\mathcal{N} (\mu,\Sigma)$ denotes that the vector $x$ is distributed as multivariate Gaussian with mean vector $\mu$ and the covariance matrix $\Sigma$, and $I_P$ is the $P\times P$ identity matrix.  
Note that, the deterministic signal can be considered as a special case of the random signal $s$ with variance $\alpha^{-1}=0$. 
Observations at the nodes are assumed to be conditionally independent and identically distributed.

Each node sends a $M$-length $(< P)$ compressed version $y_i$ of its $P$-length observation $u_i$ to the FC. The
collection of $M$-length universally sampled observations is given by, $y_i=\phi u_i$, where $\phi$ is an $M \times P$  projection matrix, which is assumed to be the same for all the nodes, and $y_i$ is the $M\times 1$ compressed observation vector (local summary statistic).  

Under the two hypotheses, the local summary statistic is 
\vspace{-0.2in}
\begin{center}
\begin{eqnarray*}
&& H_0\; :\quad y_i=\phi v_i\\
&& H_1\; :\quad y_i=\phi s+ \phi v_i.
\end{eqnarray*}
\end{center}
The FC receives compressed observation vectors, $\mathbf{y}=[y_1,\cdots,y_N]$, from the nodes via error free communication channels and makes the global decision about the phenomenon.

\subsection{Binary Hypothesis Testing at the Fusion Center}
We consider the detection problem in a Bayesian setup where the
performance criterion at the FC is the probability of error.
The FC makes the global decision about the phenomenon by considering the likelihood ratio test (LRT) which is given by

\begin{equation}
\label{test}
\prod\limits_{i=1}^{N}\dfrac{f_1(y_i)}{f_0(y_i)} \quad \mathop{\stackrel{H_1}{\gtrless}}_{H_0} \quad  \dfrac{P_0}{P_1}.
\end{equation}
Notice that, under the two hypotheses we have the following probability density functions (PDFs);

\begin{center}
\begin{eqnarray}
&& f_0(y_i)=\frac{\exp(-\frac{1}{2}y_i^T(\beta^{-1}\phi \phi^T)^{-1}y_i)}{|\beta^{-1}\phi \phi^T|^{1/2}(2\pi)^{M/2}},\label{f0}\\
&& f_1(y_i)=\frac{\exp(-\frac{1}{2}(y_i-\phi\mu)^T((\alpha^{-1}+\beta^{-1})\phi \phi^T)^{-1}(y_i-\phi\mu))}{|(\alpha^{-1}+\beta^{-1})\phi \phi^T|^{1/2}(2\pi)^{M/2}}.\label{f1}
\end{eqnarray}
\end{center}

After plugging in~\eqref{f0} and~\eqref{f1} in~\eqref{test} and taking logarithms on both sides, we obtain an equivalent test that simplifies to 

\begin{equation*}
\frac{\alpha^{-1}}{\beta^{-1}}\sum\limits_{i=1}^{N} y_i^T(\phi \phi^T)^{-1}y_i +2\sum\limits_{i=1}^{N} y_i^T(\phi \phi^T)^{-1}\mu \quad \mathop{\stackrel{H_1}{\gtrless}}_{H_0} \quad \lambda
\end{equation*}
where $\lambda=(\alpha^{-1}+\beta^{-1})\left[2 \log \frac{P_0}{P_1} +NM\log\left(1+\frac{\alpha^{-1}}{\beta^{-1}}\right)\right]+N(\phi\mu)^T(\phi\phi^T)^{-1}\phi\mu
$.

\noindent For simplicity, we assume that $P_0=P_1$. The test statistic for the collaborative compressive detector can be written in a compact form as
\begin{equation}
\label{ccd}
\Lambda(\mathbf{y})= \frac{\alpha^{-1}}{\beta^{-1}}\sum\limits_{i=1}^{N} \Lambda_1(y_i)+2\sum\limits_{i=1}^{N} \Lambda_2({y_i}) 
\end{equation}
where $\Lambda_1(y_i)=y_i^T(\phi \phi^T)^{-1}y_i$ and $\Lambda_2({y_i})=y_i^T(\phi \phi^T)^{-1}\mu$.

We would like to point out that the test statistic for the deterministic signal and random signal with zero mean cases can be seen as a special case of the above test statistic. More specifically, for the deterministic signal $s$, the test statistic is given by $\Lambda(\mathbf{y})=\sum\limits_{i=1}^{N} \Lambda_1(y_i)$ and for the zero mean random signal the test statistic is given by $\Lambda(\mathbf{y})=\sum\limits_{i=1}^{N} \Lambda_2(y_i)$ which is consistent with~\cite{bhavyaphy} and~\cite{thak}.

\section{Performance Analysis of Collaborative Compressive Detection}
\label{sec3}
First, we look at the deterministic signal case and characterize the performance of the collaborative compressive detector.

\subsection{Deterministic Signal Case}
The optimal test at the FC can be written in a compact form as 
\begin{equation*}
\sum\limits_{i=1}^{N} y_i^T(\phi \phi^T)^{-1}\phi s \quad \mathop{\stackrel{H_1}{\gtrless}}_{H_0} \quad \lambda,
\end{equation*}
with $\lambda=\frac{N}{2} s^T\phi^T(\phi \phi^T)^{-1}\phi s$. 
The decision statistic for the collaborative compressive detector is given as
\begin{equation}
\Lambda(\mathbf{y})= \sum\limits_{i=1}^{N} y_i^T(\phi \phi^T)^{-1}\phi s .
\end{equation}
We analytically characterize the performance of the collaborative compressive detector in terms of the probability of error which is defined as
\vspace{-0.15in}
\begin{center}
\begin{eqnarray*}
&& P_E=\dfrac{1}{2} P_F+\dfrac{1}{2} (1-P_D)
\end{eqnarray*}
\end{center}
where, $P_F=P(\Lambda(\mathbf{y})>\lambda|H_0)$ and $P_D=P(\Lambda(\mathbf{y})>\lambda|H_1)$
is the probability of false alarm and the probability of detection, respectively.
To simplify the notations, we define 

\begin{equation*}
\hat{P}= \phi^T(\phi \phi^T)^{-1}\phi
\end{equation*}
as the orthogonal projection operator onto row space of $\phi$.
Using this notation, it is easy to show that

\[ \Lambda(\mathbf{y}) \sim \left\{ \begin{array}{rll}
				\mathcal{N}(0,\beta^{-1}N\|\hat{P}s\|_2^2),  & \mbox{under}\ H_{0} \\
			    \mathcal{N}(N\|\hat{P}s\|_2^2,\beta^{-1}N\|\hat{P}s\|_2^2)  & \mbox{under}\ H_{1}
				\end{array}\right. 
\] 
where $\|\hat{P}s\|_2^2=s^T\phi^T(\phi \phi^T)^{-1}\phi s$.

Thus, we have
\begin{equation}
P_F = Q \left(\frac{\frac{N}{2}\|\hat{P}s\|_2^2}{ \sqrt{N\beta^{-1}}\|\hat{P}s\|_2}\right) 
= Q\left(\frac{1}{2}\sqrt{\frac{N}{\beta^{-1}}}\|\hat{P}s\|_2\right)
\end{equation}
and
\begin{equation}
P_D= Q \left(\frac{\frac{N}{2}\|\hat{P}s\|_2^2-N\|\hat{P}s\|_2^2}{ \sqrt{N\beta^{-1}}\|\hat{P}s\|_2}\right)  
= Q\left(-\frac{1}{2}\sqrt{\frac{N}{\beta^{-1}}}\|\hat{P}s\|_2\right)
\end{equation}
where $Q(x)=\frac{1}{\sqrt{2\pi}}\int_{x}^{\infty}\exp(-\frac{u^2}{2})\;du$.

\noindent The probability of error can be calculated to be

\begin{equation}
\label{pe}
P_E\;=\;Q\left(\frac{1}{2}\sqrt{\frac{N}{\beta^{-1}}}\|\hat{P}s\|_2\right).
\end{equation} 

Next, we derive the modified deflection coefficient (first proposed in~\cite{MDC}) of the system and show its monotonic relationship with the probability of error as given in~\eqref{pe}. The modified deflection coefficient provides a good measure of the detection performance
since it characterizes the variance-normalized distance between
the centers of two conditional PDFs.
Notice that, for the deterministic signal case, $y_i$ is distributed under the hypothesis $H_j$ as, $y_i\sim \mathcal{N}(\mu_j^i,\Sigma_j^i)$. The modified deflection coefficient $D(\mathbf{y})$ can be obtained to be
\begin{eqnarray*}
D(\mathbf{y})&=& \sum\limits_{i=1}^{N} (\mu_1^i-\mu_0^i)^T (\Sigma_1^i)^{-1} (\mu_1^i-\mu_0^i)\\
&=& N\frac{\|\hat{P}s\|_2^2}{\beta^{-1}}.
\end{eqnarray*}

The monotonic relationship between $P_E$ as given in~\eqref{pe} and $D(\mathbf{y})$ can be observed by noticing that
\begin{equation*}
P_E=Q\left(\frac{\sqrt{D(\mathbf{y})}}{2}\right).
\end{equation*}
Later in the paper, we will use the modified deflection coefficient to characterize the detection performance of the system. 

Notice that, the detection performance is a function of the projection operator $\hat{P}$. In general, this performance could be either quite good or quite poor depending on the random projection matrix $\phi$. 
Next, we provide bounds on the performance of the collaborative compressive detector using the concept of $\epsilon$-stable embedding.\footnote{To construct
linear mappings that satisfy an $\epsilon$-stable embedding property is beyond the scope of this work. We refer interested readers to~\cite{dave}.}

\begin{definition}
\label{stableembedding}
Let $\epsilon\in(0,1)$ and $\mathcal{S},\mathcal{X}\subset\mathbb{R}^P$. We say that a mapping $\psi$ is an $\epsilon$-stable embedding of $(\mathcal{S},\mathcal{X})$ if
\begin{equation*}
(1-\epsilon)\; \|s-x\|_2^2 \leq \|\psi s-\psi x\|_2^2\leq (1+\epsilon)\; \|s-x\|_2^2,
\end{equation*}
for all $s\in \mathcal{S}$ and $x\in \mathcal{X}$.
\end{definition}

Using this concept, we state our result in the next theorem.
\begin{theorem}
\label{th1}
Suppose that $\sqrt{\frac{P}{M}}\hat{P}$ provides an $\epsilon$-stable embedding of $(\mathcal{S},\{0\})$. Then for any deterministic signal $s\in S$, the probability of error of the collaborative compressive detector satisfies

\begin{small}
\begin{equation*}
Q\left(\sqrt{1+\epsilon}\frac{\sqrt{N}}{2}\sqrt{\frac{M}{P}}\frac{\|s\|_2}{\sqrt{\beta^{-1}}}\right)\leq P_E\leq Q\left(\sqrt{1-\epsilon}\frac{\sqrt{N}}{2}\sqrt{\frac{M}{P}}\frac{\|s\|_2}{\sqrt{\beta^{-1}}}\right).
\end{equation*} 
\end{small}
\end{theorem}
\begin{proof}
By our assumption that $\sqrt{\frac{P}{M}}\hat{P}$ provides an
$\epsilon$-stable embedding of $(\mathcal{S},\{0\})$, we know that
\begin{equation}
\label{c1}
\sqrt{1-\epsilon}\; \|s\|_2 \leq \sqrt{\frac{P}{M}}\;\|\hat{P}s\|_2\leq \sqrt{1+\epsilon}\; \|s\|_2.
\end{equation}
Combining~\eqref{c1} with~\eqref{pe}, the result follows. 
\end{proof}

For nominal values of $\epsilon$, $P_E$ can be approximated as 

\begin{equation*}
P_E \approx Q\left(\frac{\sqrt{N}}{2}\sqrt{\frac{M}{P}}\frac{\|s\|_2}{\sqrt{\beta^{-1}}}\right).
\end{equation*}

The above expression tells us in a precise way how much information we lose by using low dimensional projections rather than the signal samples themselves. It also tells us how many nodes are needed to collaborate to compensate for the loss due to compression. More specifically, if $N\geq \frac{1}{c}$, where $c=\frac{M}{P}$ is defined as the compression ratio at each node, the loss due to compression can be recovered.
Notice that, for a fixed $M$, as the number of collaborating nodes approaches infinity, i.e., $N\rightarrow\infty$, the probability of error vanishes. On the other hand, to guarantee $P_E\leq \delta$, parameters $M,\;P$ and $N$ should satisfy 

\begin{equation*}
cN\geq \frac{4}{SNR}(Q^{-1}(\delta))^2 
\end{equation*}
where $SNR=\frac{\|s\|_2^2}{\beta^{-1}}$.
\begin{figure}[!t]
\centering
    \includegraphics[width=3.5in]{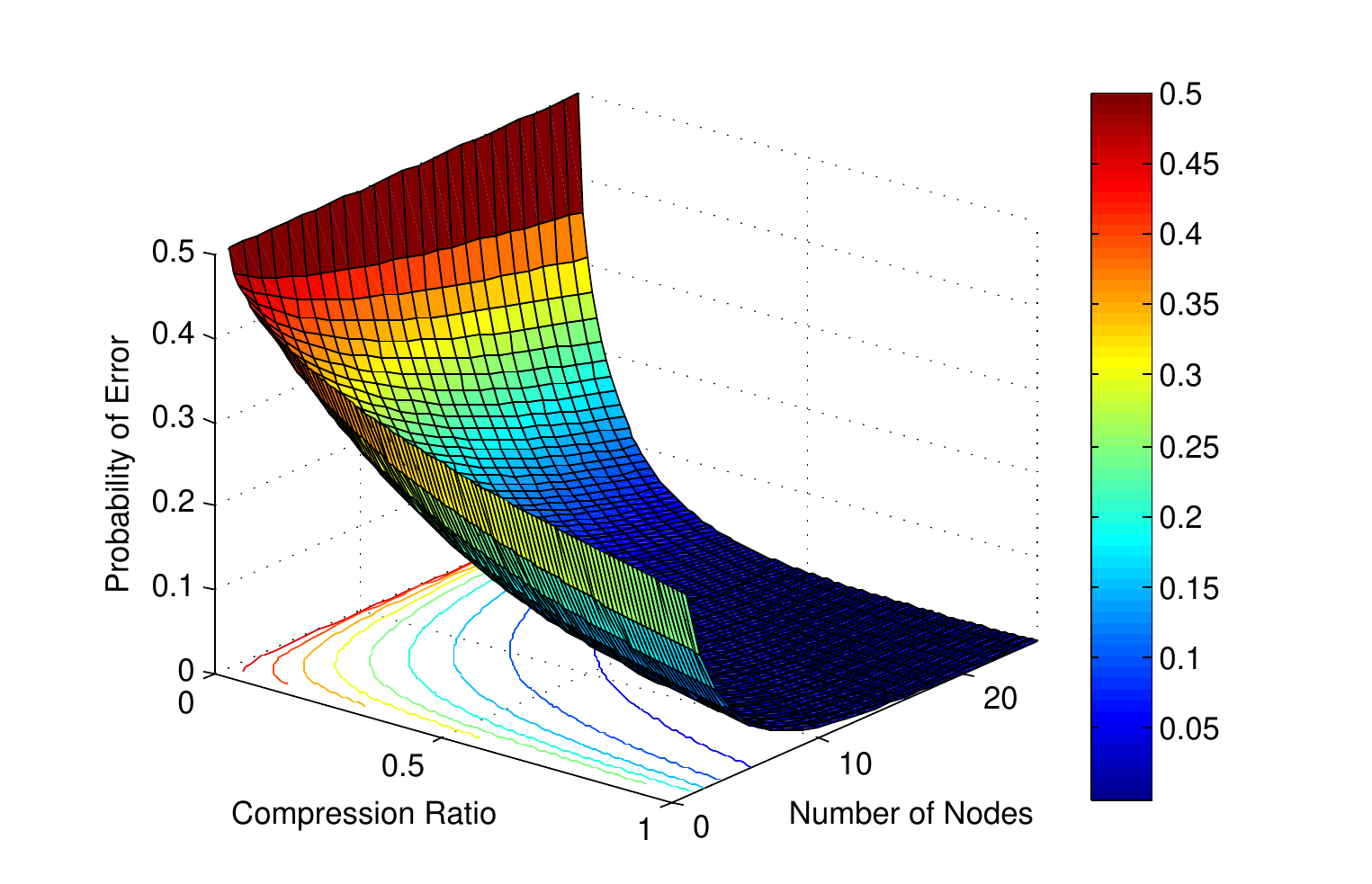}
\vspace{-0.2in}
    \caption{Prob. of error as a function of number of nodes and compression ratio $c=M/P$ for $SNR=3dB$}
    \label{collvscomp}
\end{figure}

To corroborate our theoretical results, in Figure~\ref{collvscomp} we present the behavior of $P_E$ with respect of collaboration and compression. We plot $P_E$ as a function of the number of nodes $N$ and compression ratio $c$. We assume that $SNR=3dB$. It can be seen from Figure~\ref{collvscomp} that $P_E$ is a monotonically decreasing function of $c$ and $N$, and, therefore, the performance loss due to compression can be compensated by exploiting spatial diversity or collaboration.
  
In order to more clearly illustrate the behavior of $P_E$ with respect to compression and collaboration, we also establish the following corollary of Theorem~\ref{th1} using the Chernoff Bound.
\begin{corollary}
\label{c2}
Suppose that $\sqrt{\frac{P}{M}}\hat{P}$ provides an $\epsilon$-stable embedding of $(\mathcal{S},\{0\})$. Then for any deterministic signal $s\in S$, we have

\begin{equation*}
P_E \leq \frac{1}{2}\exp\left(-\frac{1}{8}cN\frac{\|s\|_2^2}{\beta^{-1}}\right).
\end{equation*}
\end{corollary}

Corollary~\ref{c2} suggests that the error probability vanishes exponentially fast as we increase either the compression ratio $c$ or the number of collaborating nodes $N$. 

Next, we extend the above analysis to the case where the signal of interest is a random signal such that $s\sim\mathcal{N}(\mu,\alpha^{-1}I_P)$.

\subsection{Random Signal with Arbitrary Mean Case}
Let the signal of interest be $s \sim \mathcal{N} (\mu,\alpha^{-1} I_P)$ with an arbitrary $\mu$. Then, the collaborative compressive detector is given by
\begin{equation*}
\frac{\alpha^{-1}}{\beta^{-1}}\sum\limits_{i=1}^{N} y_i^T(\phi \phi^T)^{-1}y_i +2\sum\limits_{i=1}^{N} y_i^T(\phi \phi^T)^{-1}\mu \quad \mathop{\stackrel{H_1}{\gtrless}}_{H_0} \quad \lambda
\end{equation*}
where $\lambda=(\alpha^{-1}+\beta^{-1})\left[NM\log\left(1+\frac{\alpha^{-1}}{\beta^{-1}}\right)\right]+N(\phi\mu)^T(\phi\phi^T)^{-1}\phi\mu
$.
Note that, the test statistic is of the form $\sum\limits_{i=1}^{N}\left[y_i^TAy_i+2b^Ty_i\right]$ with $A=\frac{\alpha^{-1}}{\beta^{-1}}(\phi\phi^T)^{-1}$ and $b=(\phi\phi^T)^{-1}\phi\mu$. 
In general, it is difficult to find the PDF of such an expression in a closed form. Next, we state a Lemma from~\cite{ogasawara}, which will be used to derive the distribution of the test statistic in a closed form. 
\begin{lemma}[\cite{ogasawara}]
\label{ogaswar}
Let $A$ be a symmetric matrix and $x\sim\mathcal{N}(\mu,V)$, where $V$ is positive definite (hence nonsingular). The necessary and sufficient condition that $x^TAx+2b^Tx+c$ follows a noncentral chi-squared distribution $\mathcal{X}^2_{k}(\delta)$ with $k$ degrees of freedom and noncentrality parameter $\delta$ is that

\begin{equation}
\label{oseq}
\begin{bmatrix}
    A        \\
    \hdotsfor{1} \\
    b^T 
\end{bmatrix}
V [A:b]=
\begin{bmatrix}
    A & b \\
    b^T & c
\end{bmatrix}
\end{equation}

in which case $k$ is the rank of $A$ and $\delta=\mu^TA\mu+2b^T\mu+c$. 
\end{lemma}

Next, using Lemma~\ref{ogaswar} we state the following proposition.
\begin{proposition}
\label{pro1}
For a $P\times P$ symmetric and idempotent matrix $S$ and $u_i\sim\mathcal{N}(\mu,\sigma^2I_P)$, the test statistic of the form $u_i^TAu_i+2b^Tu_i+c$ with $A=\frac{1}{\sigma^2}S$, $b^T=\frac{1}{\sigma^2}z^TS$ and $c=\frac{1}{\sigma^2}z^TSz$ follows a noncentral chi-squared distribution $\mathcal{X}^2_{k}(\delta)$ where $k=\text{Rank}(S)$ and $\delta=\frac{1}{\sigma^2}\left(\mu^TS\mu+2z^TS\mu+z^TSz\right)$ for any arbitrary $P\times 1$ vector $z$.
\end{proposition}
\begin{proof}
To prove the proposition, it is sufficient to show that the above mentioned $A$, $b$ and $c$ satisfy condition~\eqref{oseq} in Lemma~\ref{ogaswar} for any arbitrary $P\times 1$ vector $z$. Notice that, $S$ satisfies the following properties: symmetric $S^T=S$ and idempotent $S^2=S$. Thus, 
\begin{eqnarray*}
\begin{bmatrix}
    A        \\
    \hdotsfor{1} \\
    b^T 
\end{bmatrix}
V [A:b]&=&
\begin{bmatrix}
    \frac{1}{\sigma^2}S        \\
    \hdotsfor{1} \\
    \frac{1}{\sigma^2}z^TS 
\end{bmatrix}\left[\sigma^2I_P\right] \left[\frac{1}{\sigma^2}S\;\vdots\;\frac{1}{\sigma^2}Sz \right]
\\
&=&
\begin{bmatrix}
    \hat{P}        \\
    \hdotsfor{1} \\
    z^TS
\end{bmatrix} \left[\frac{1}{\sigma^2}S\;\vdots\;\frac{1}{\sigma^2}Sz \right]
\\
&=&
\begin{bmatrix}
    \frac{1}{\sigma^2}SS  &  \frac{1}{\sigma^2}SSz   \\
    \frac{1}{\sigma^2}z^TSS & \frac{1}{\sigma^2}z^TSS z
\end{bmatrix}  
\\
&=&
\begin{bmatrix}
    A & b \\
    b^T & c
\end{bmatrix}
\end{eqnarray*}
Thus, the test statistic follows a noncentral chi-squared distribution $\mathcal{X}^2_{k}(\delta)$ where $k=\text{Rank}(S)$ and $\delta=\frac{1}{\sigma^2}\left(\mu^TS\mu+2z^TS\mu+z^TSz\right)$ for any arbitrary $z$.
\end{proof}

Using these results, the collaborative compressive detector reduces to:
\begin{equation*}
\frac{\alpha^{-1}}{\beta^{-1}}\sum\limits_{i=1}^{N} y_i^T(\phi \phi^T)^{-1}y_i +2\sum\limits_{i=1}^{N} y_i^T(\phi \phi^T)^{-1}\mu \quad \mathop{\stackrel{H_1}{\gtrless}}_{H_0} \quad \lambda.
\end{equation*}
Using the fact that $y_i=\phi u_i$ and rearranging the terms, we get
\begin{equation*}
\sum\limits_{i=1}^{N} \left[u_i^T\hat{P}u_i +2\frac{\beta^{-1}}{\alpha^{-1}} \mu^T\hat{P}u_i+\left(\frac{\beta^{-1}}{\alpha^{-1}}\right)^2\mu^T\hat{P}\mu\right] \quad \mathop{\stackrel{H_1}{\gtrless}}_{H_0} \quad \tau
\end{equation*}
 where $\hat{P}=\phi^T(\phi\phi^T)^{-1}\phi$ and $\tau=\frac{\beta^{-1}}{\alpha^{-1}}\lambda+N\left(\frac{\beta^{-1}}{\alpha^{-1}}\right)^2\mu^T\hat{P}\mu$.
Note that, the FC does not have access to $u_i$ and the above test statistic is used only for deriving the PDF of the original test statistic.

\begin{figure*}[t]
\centering
\subfigure[]{
\includegraphics[%
  width=0.4\textwidth,clip=true]{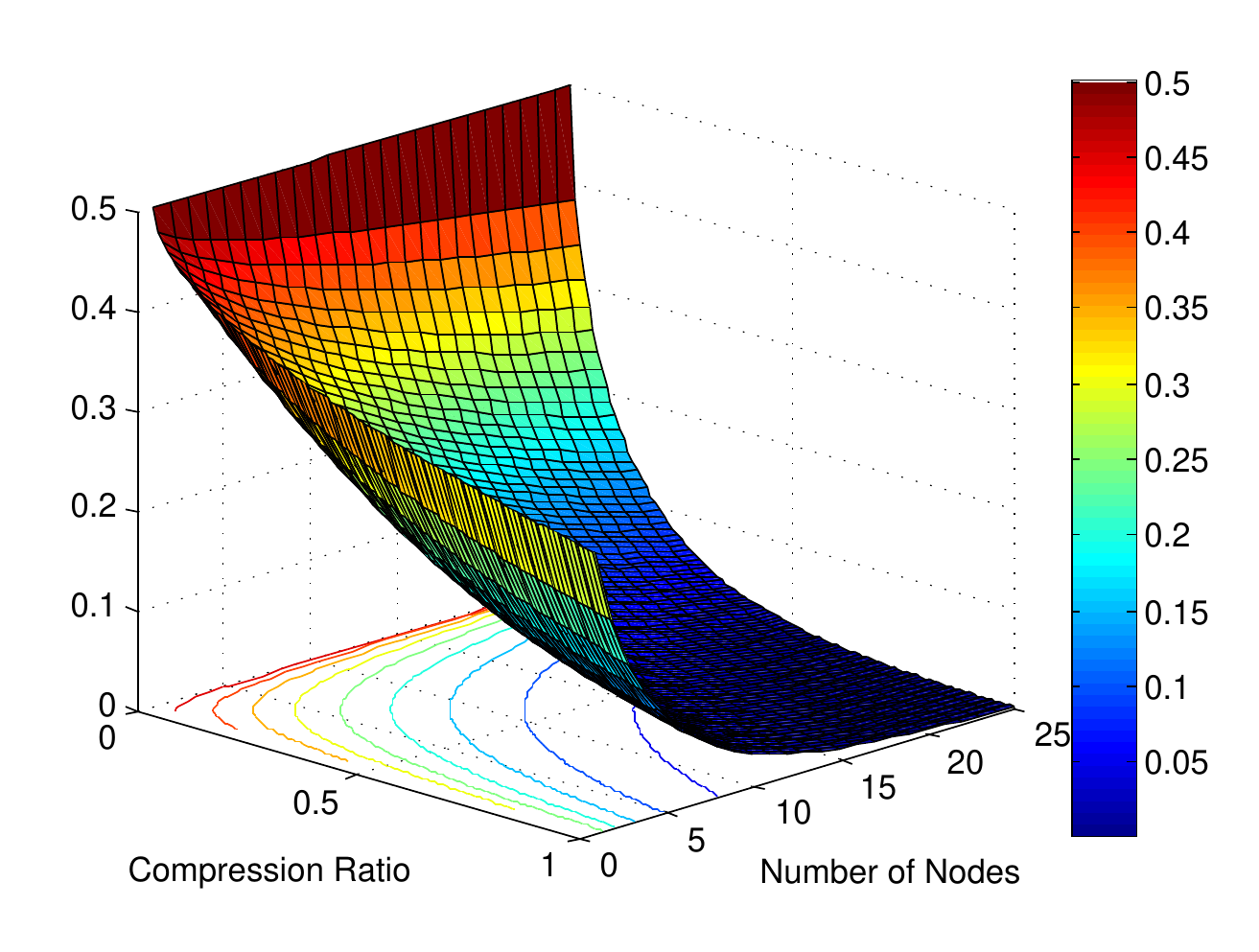}
\label{PeR}}
\hspace{0.1in}
\subfigure[] {
\includegraphics[%
  width=0.4\textwidth,clip=true]{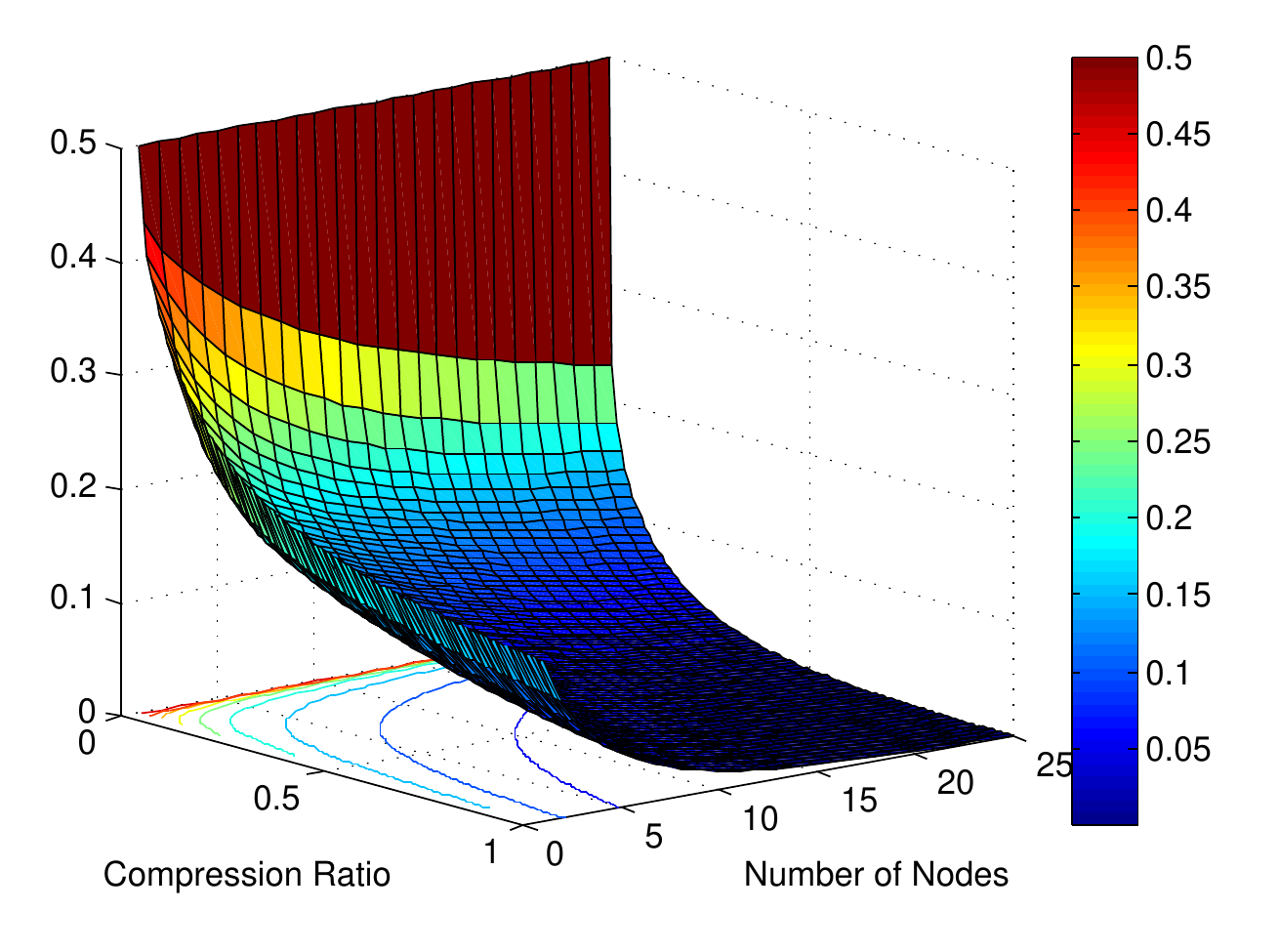}
\label{PeRm} }
\caption{Prob of error ($P_e$) analysis when $(\alpha^{-1},\beta^{-1})=(1,20)$ and $P=100$.  \subref{PeR} $P_e$ with varying ($(c,\;N)$) when $\mu=0$. \subref{PeRm} $P_e$ with varying ($(c,\;N)$) when $\mu=10^{-3}$.}
\label{PeRandom}
\end{figure*}

\begin{theorem}
\label{therm}
For a projection matrix $\hat{P}=\phi^T(\phi\phi^T)^{-1}\phi$ and $u_i\sim\mathcal{N}(\mu,\sigma_k^2I_P)$ under the hypothesis $H_k$, the test statistic 
$$\Lambda(\mathbf{y})=\sum\limits_{i=1}^{N} \left[u_i^T\hat{P}u_i +2\frac{\beta^{-1}}{\alpha^{-1}} \mu^T\hat{P}u_i+\left(\frac{\beta^{-1}}{\alpha^{-1}}\right)^2\mu^T\hat{P}\mu\right]$$
has the following distribution 
\[\frac{\Lambda(\mathbf{y})}{\sigma_k^2} \sim \left\{\begin{array}{rll}
				\mathcal{X}^2_{NM}(\delta_0),  & \mbox{under}\ H_{0} \\
				\mathcal{X}^2_{NM}(N\delta_1)  & \mbox{under}\ H_{1}
				\end{array}\right. 
\] 
where $\sigma_0^2=\beta^{-1}$, $\sigma_1^2=\alpha^{-1}+\beta^{-1}$ and $\mathcal{X}^2_{NM}(\delta_k)$ denotes noncentral chi-square distribution with $NM$ degrees of freedom and parameters $\delta_0=0$ and $\delta_1=\frac{||\hat{P}\mu||_2^2}{\alpha^{-1}}\left(1+\frac{\beta^{-1}}{\alpha^{-1}}\right)$. 
\end{theorem}
\begin{proof}
Let us denote the test static by $\Lambda(\mathbf{y})=\sum\limits_{i=1}^{N}\Lambda({y_i})$ with 

$$\Lambda({y_i})=\left[u_i^T\hat{P}u_i +2\frac{\beta^{-1}}{\alpha^{-1}} \mu^T\hat{P}u_i+\left(\frac{\beta^{-1}}{\alpha^{-1}}\right)^2\mu^T\hat{P}\mu\right].$$ 

Now notice that, $\frac{\Lambda({y_i})}{\sigma_k^2}$ is of the form $u_i^TAu_i+2b^Tu_i+c$ with $A=\frac{\hat{P}}{\sigma_k^2}$, $b^T=\frac{z^T\hat{P}}{\sigma_k^2}$, $z=\frac{\beta^{-1}}{\alpha^{-1}}\mu$ and $c=\frac{z^T\hat{P}z}{\sigma_k^2}$. Also note that, the projection matrix $\hat{P}=\phi^T(\phi\phi^T)^{-1}\phi$ is both symmetric and idempotent with rank$(\hat{P})=M$. As a result, using Proposition~\ref{pro1} and the fact that $\frac{\Lambda(\mathbf{y})}{\sigma_k^2}$ is the sum of $N$ I.I.D. chi-squared random variables $\frac{\Lambda(y_i)}{\sigma_k^2}$, the result in the Theorem~\ref{therm} can be derived.
\end{proof}

If $NM$ is large enough, then, the following approximations hold

\[ \frac{\Lambda(\mathbf{y})}{\sigma_k^2} \sim \left\{\begin{array}{rll}
				\mathcal{N}(NM,2NM),  & \mbox{under}\ H_{0} \\
			    \mathcal{N}((NM+N\delta_1),2(NM+N\delta_1)  & \mbox{under}\ H_{1}
				\end{array}\right. 
\]
where $\delta_1=\frac{||\hat{P}\mu||_2^2}{\alpha^{-1}}\left(1+\frac{\beta^{-1}}{\alpha^{-1}}\right)$. 
As a result, we have
$$\qquad P_F=P\left(\frac{\Lambda(\mathbf{y})}{\beta^{-1}}>\frac{\tau}{\beta^{-1}}|H_0\right)=Q\left(\frac{\frac{\tau}{\beta^{-1}}-NM}{\sqrt{2NM}}\right)$$ and
\begin{eqnarray*}
P_D&=& P\left(\frac{\Lambda(\mathbf{y})}{\alpha^{-1}+\beta^{-1}}>\frac{\tau}{\alpha^{-1}+\beta^{-1}}|H_1\right)\\
&=& Q\left(\frac{\frac{\tau}{\alpha^{-1}+\beta^{-1}}-NM-N\delta_1}{\sqrt{2(NM+N\delta_1)}}\right)
\end{eqnarray*}
where $\tau=\frac{\beta^{-1}}{\alpha^{-1}}\lambda+N\left(\frac{\beta^{-1}}{\alpha^{-1}}\right)^2||\hat{P}\mu||_2^2$, $\lambda=(\alpha^{-1}+\beta^{-1})\left[NM\log\left(1+\frac{\alpha^{-1}}{\beta^{-1}}\right)\right]+N||\hat{P}\mu||_2^2$ and $\delta_1=\frac{||\hat{P}\mu||_2^2}{\alpha^{-1}}\left(1+\frac{\beta^{-1}}{\alpha^{-1}}\right)$. 

The detection performance of the system is a function of the projection operator $\hat{P}$.  
Next, we provide approximations to the performance of the collaborative compressive detector using the concept of $\epsilon$-stable embedding of the mean $\mu$.

\begin{theorem}
\label{thr}
Suppose that $\sqrt{\frac{P}{M}}\hat{P}$ provides an $\epsilon$-stable embedding of $(\mathcal{U},\{0\})$. Then for any random signal $s\sim\mathcal{N}(\mu,\alpha^{-1}I_P)$ with $\mu\in \mathcal{U}$, the probability of error of the collaborative compressive detector can be approximated as
\begin{equation*}
P_E= \dfrac{1}{2} Q\left(\sqrt{cN}\tau_0\right)+\dfrac{1}{2} \left(Q\left(\sqrt{cN}\tau_1\right)\right)
\end{equation*}
where $\tau_0=\sqrt{\frac{P}{2}}\left(\left(1+\tau^{-1}\right)\left(\log\left(1+\tau\right)+\frac{\|\mu\|_2^2}{\alpha^{-1}P}\right)-1\right)$, $\tau_1=\sqrt{\frac{P+\delta_1'}{2}}\left(1-\frac{\tau^{-1}\left(P\log(1+\tau)+\frac{\|\mu\|_2^2}{\alpha^{-1}}\right)}{P+\delta_1'}\right)$ with $\delta_1'=\frac{\|\mu\|_2^2}{\alpha^{-1}}(1+\tau^{-1})$ and $\tau=\frac{\alpha^{-1}}{\beta^{-1}}$.
\end{theorem}
\begin{proof}
By our assumption that $\sqrt{\frac{P}{M}}\hat{P}$ provides an
$\epsilon$-stable embedding of $(\mathcal{U},\{0\})$, we know that
\begin{equation}
\sqrt{1-\epsilon}\; \|\mu\|_2 \leq \sqrt{\frac{P}{M}}\;\|\hat{P}\mu\|_2\leq \sqrt{1+\epsilon}\; \|\mu\|_2.
\end{equation}
In other words, for large values of $NM$ the following approximation holds: $|\hat{P}\mu\|_2^2\approx\frac{M}{P}\|\mu\|_2^2=c\|\mu\|_2^2$. 
The proof follows from the fact that $Q(x)=1-Q(-x)$ and by plugging in
\begin{equation}
P_F = Q\left(\sqrt{cN}\tau_0\right) 
\end{equation}
and
\begin{equation}
P_D = Q\left(-\sqrt{cN}\tau_1\right)
\end{equation}
in the equation $P_E=\dfrac{1}{2} P_F+\dfrac{1}{2} (1-P_D)$,
the above results can be derived.
\end{proof}
Note that, $(-\tau_1)\leq\tau_0$ and, therefore, $P_D\geq P_F$. Similar to the deterministic signal case, if $N\geq c^{-1}$ the loss due to compression with a single node can be recovered in collaborative compressive detection for the random signal case as well.
For a fixed $M$, as the number of collaborating nodes approaches infinity, i.e., $N\rightarrow\infty$, the probability of error vanishes. We would like to point out that by plugging in $\mu=0$ in the above expressions, results for the zero mean signal case can be derived (which are consistent with~\cite{thak,gen}). 

To gain insights into Theorem~\ref{thr}, we present illustrative examples that corroborate our results. In Figure~\ref{PeR} we plot the probability of error $P_E$ as a function of the number of nodes $N$ and compression ratio $c$. We assume that the signal of interest is $s\sim\mathcal{N}(0,I_P)$ and noise $v_i\sim\mathcal{N}(0,20 I_P)$, with original length of the signal being $P=100$. It can be seen from the figure that $P_E$ is a monotonically decreasing function of $(c,N)$. 
In Figure~\ref{PeRm}, we plot the probability of error $P_E$ as a function of $(c,N)$ when the signal of interest is $s\sim\mathcal{N}(\mu,I_P)$ with $\|\mu\|_2^2=10^{-3}$. Similar to Figure~\ref{PeRm}, $P_E$ decreases monotonically with $(c,N)$, however, with a much faster rate. 
In order to more clearly illustrate this behavior of $P_E$, we also establish the following corollary of Theorem~\ref{thr}.

\begin{corollary}
\label{cr}
Suppose that $\sqrt{\frac{P}{M}}\hat{P}$ provides an $\epsilon$-stable embedding of $(\mathcal{U},\{0\})$. Then for any random signal $s\sim\mathcal{N}(\mu,\alpha^{-1}I_P)$ with $\mu\in \mathcal{U}$, the error probability $P_E$ of the collaborative compressive detector satisfies 
\begin{equation*}
P_E \leq \frac{1}{4}\exp\left(-\frac{cN}{2}\tau_0^2\right)+\frac{1}{4}\exp\left(-\frac{cN}{2}\tau_1^2\right).
\end{equation*}
\end{corollary}
\begin{proof}
To prove the corollary, we first show that both $\tau_0$ and $\tau_1$ as given in Theorem~\ref{thr} are positive. Let us  denote by $\tau_k(\mu=0)$ the expression when $\mu=0$ is plugged in the expression for $\tau_k$ for $k\in\{0,1\}$. Then, it can be shown that $\tau_k\geq \tau_k(\mu=0)$ for $k\in\{0,1\}$. Now, a sufficient condition for $\tau_k>0$ is $\tau_k(\mu=0)>0$ for $k\in\{0,1\}$.
The condition for $\tau_0(\mu=0)>0$ and $\tau_1(\mu=0)>0$ to be true can be written as
\begin{equation*}
\frac{1}{1+\frac{\beta^{-1}}{\alpha^{-1}}}<\log\left(1+\frac{\alpha^{-1}}{\beta^{-1}}\right)<\frac{\alpha^{-1}}{\beta^{-1}}.
\end{equation*}
The above condition can be shown to be true by applying the logarithm inequality $\frac{\tau}{1+\tau}<\log(1+\tau)<\tau$ with $\tau=\frac{\alpha^{-1}}{\beta^{-1}}$.
Now using the Chernoff bound (i.e., $Q(x)\leq\frac{1}{2}\exp(-\frac{x^2}{2})$ for $x>0$), it can be shown that
\begin{equation*}
P_E \leq \frac{1}{4}\exp\left(-\frac{cN}{2}\tau_0^2\right)+\frac{1}{4}\exp\left(-\frac{cN}{2}\tau_1^2\right).
\end{equation*}
The above expression suggests that $P_E$ vanishes exponentially fast as we increase either the compression ratio $c$ or the number of collaborating nodes $N$. 
\end{proof}

Next, we consider the problem where the network operates in the presence of an eavesdropper who wants to discover the state of the nature being monitored by the system.
The FC's goal is to implement the appropriate countermeasures to keep  the data regarding the presence of the phenomenon secret from the eavesdropper. 

\section{Collaborative Compressive Detection in the Presence of an Eavesdropper}
\label{sec4}
In a collaborative compressive detection framework, the FC receives compressed observation vectors, $\mathbf{y}=[y_1,\cdots,y_N]$, from the nodes and makes the global decision about the presence of the random signal vector\footnote{In rest of the paper, we will consider only the random signal detection case. Deterministic signal can be seen as a special case of random signal $s$ with variance $\alpha^{-1}=0$ and results for the deterministic signal case can be obtained by plugging in $\alpha^{-1}=0$ in corresponding expressions for the random signal case.} $s \sim \mathcal{N} (\mu,\alpha^{-1} I_P)$ with $\mu\neq 0$. The transmission of the
nodes, however, may be observed by an eavesdropper who also wants to discover the state of the phenomenon (see Figure~\ref{Fig: modelE}). To keep the data regarding the presence of the phenomenon secret from the eavesdropper, we propose to use cooperating trustworthy nodes that assist the FC by injecting artificial noise to mislead the eavesdroppers to improve the security performance of the system. 

\begin{figure}[!t]
\centering
    \includegraphics[width=2.75in]{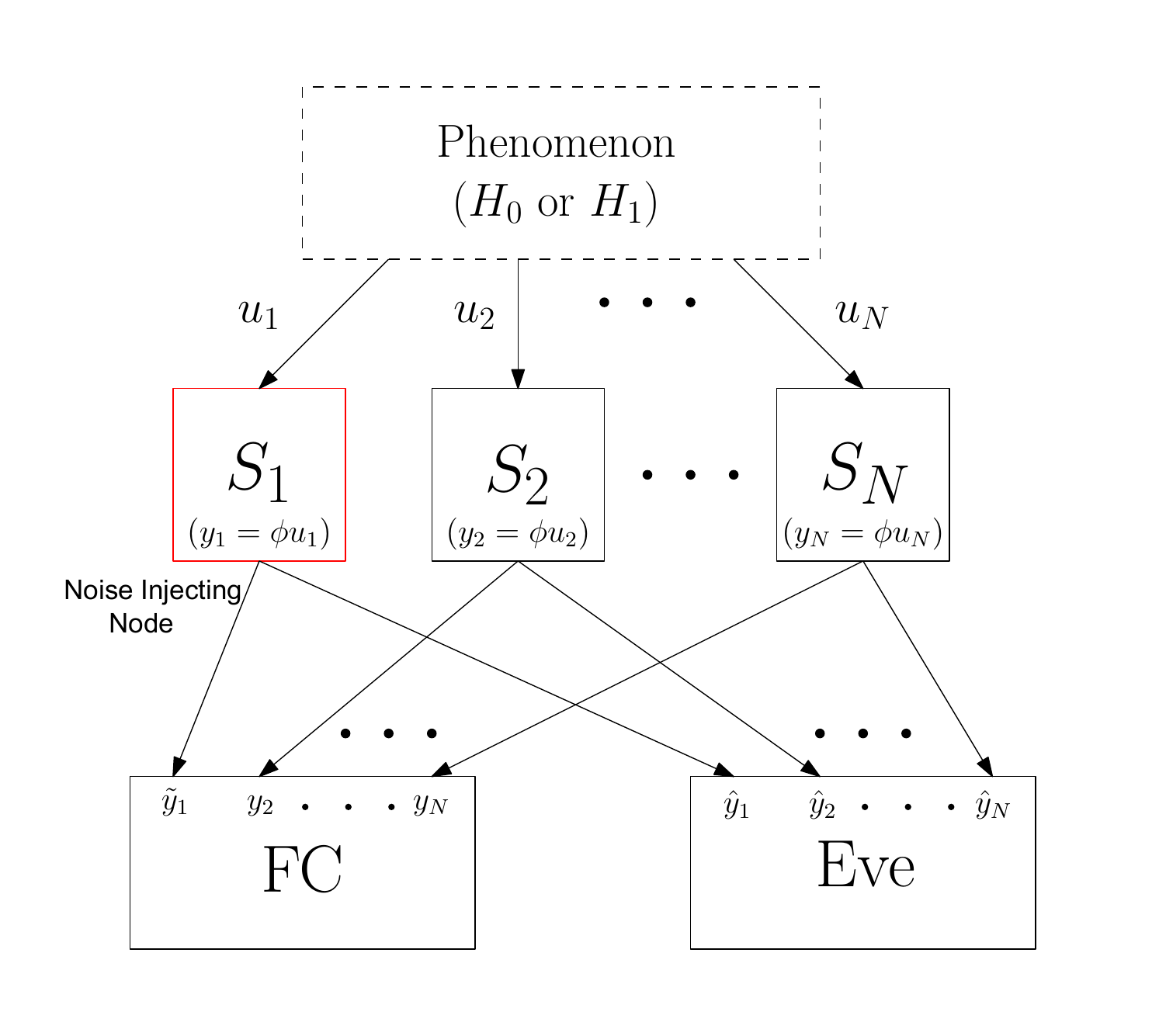}
\vspace{-0.2in}
    \caption{Collaborative Compressive Detection Network in the Presence of an Eavesdropper}
    \label{Fig: modelE}
\end{figure}

\subsection{Artificial Noise Injection Model}
It is assumed that $B$ out of $N$ nodes (or $\alpha$ fraction of the nodes) inject artificial noise according to the model given next.
Nodes tamper their data $y_i$ and send $\tilde{y_i}$ in the following manner: 

\noindent Under $H_0$: 

\[ \tilde{y_{i}} =  \left\{ \begin{array}{rll}
				 \phi(v_{i}+W_{i})  & \mbox{with probability}\ P_{1}^0 \\
				 \phi(v_{i}-W_{i})  & \mbox{with probability}\ P_{2}^0 \\
			     \phi v_{i}  & \mbox{with probability}\ (1-P_1^0-P_2^0)\\
				\end{array}\right.  
\]

\noindent Under $H_1$: 

\[ \tilde{y_{i}} =  \left\{ \begin{array}{rll}
				 \phi(s+v_{i}+W_{i})  & \mbox{with probability}\ P_{1}^1 \\
				 \phi(s+v_{i}-W_{i})  & \mbox{with probability}\ P_{2}^1 \\
			     \phi(s+v_{i})  & \mbox{with probability}\ (1-P_1^1-P_2^1)\\
				\end{array}\right.  
\]
\vspace{0.075in}
where the signal $s$ is assumed to be distributed as $s \sim \mathcal{N}(\mu,\alpha^{-1}I_P)$ and
$W_{i}$ is the artificial noise injected in the system which is distributed as AWGN $W_i\sim \mathcal{N}(D_i,\gamma^{-1}I_P)$ with $D_i=\kappa\mu$. The parameter $\kappa>0$ represents the artificial noise strength, which is zero for non artificial noise injecting nodes. 
Also note that, the values of $(P_1^0, P_2^0)$ and $(P_1^1, P_2^1)$ are system dependent. For example, under the assumption that the noise injecting nodes have perfect knowledge of the hypothesis, we have $P_1^0=1$ and $P_2^1=1$. In other scenarios, values of $(P_1^0, P_2^0)$ and $(P_1^1, P_2^1)$ are constrained by the local detection capability of the nodes. However, it is reasonable to assume that $(P_1^0 > P_2^0)$ and $(P_1^1<P_2^1)$ because under hypothesis $H_0$ the tampered value should be high and under $H_1$ the tampered value should be low to degrade the performance at the eavesdropper.
We assume that the observation model and artificial noise parameters (i.e., $\kappa\; \text{and}\; \gamma^{-1}$) are known to both the FC and the eavesdropper. The only information unavailable at the eavesdropper is the identity of the noise injecting nodes (Byzantines) and considers each node $i$ to be Byzantine with a certain probability $\alpha$. 

\subsection{Binary Hypothesis Testing}

The FC can distinguish between $y_i$ and $\tilde{y_i}$. Notice that, $\tilde{y_i}$ is distributed under the hypothesis $H_0$ as a multivariate Gaussian mixture $\mathcal{N}(P_k^0,\tilde{\mu_0}^i,\tilde{\Sigma_0}^i)$ which comes from $\mathcal{N}(\phi D_i,(\gamma^{-1}+\beta^{-1})\phi\phi^T)$ with probability $P_1^0$, from $\mathcal{N}(-\phi D_i,(\gamma^{-1}+\beta^{-1})\phi\phi^T)$ with probability $P_2^0$ and from $\mathcal{N}(0,(\gamma^{-1}+\beta^{-1})\phi\phi^T)$ with probability $(1-P_1^0-P_2^0)$. Similarly, under the hypothesis $H_1$ it is distributed as multivariate Gaussian mixture $\mathcal{N}(P_k^1,\tilde{\mu_1}^i,\tilde{\Sigma_1}^i)$ which comes from $\mathcal{N}(\phi(\mu+D_i),(\alpha^{-1}+\gamma^{-1}+\beta^{-1})\phi\phi^T)$ with probability $P_1^1$, from $\mathcal{N}(\phi(\mu- D_i),(\alpha^{-1}+\gamma^{-1}+\beta^{-1})\phi\phi^T)$ with probability $P_2^1$ and from $\mathcal{N}(\phi \mu,(\alpha^{-1}+\gamma^{-1}+\beta^{-1})\phi\phi^T)$ with probability $(1-P_1^1-P_2^1)$. The FC makes the global decision about the phenomenon by considering the likelihood ratio test (LRT) which is given by
\begin{equation}
\label{l1}
\prod\limits_{i=1}^{B}\dfrac{f_1(\tilde{y_i})}{f_0(\tilde{y_i})}\prod\limits_{i=B+1}^{N}\dfrac{f_1(y_i)}{f_0(y_i)} \quad \mathop{\stackrel{H_1}{\gtrless}}_{H_0} \quad  \dfrac{P_0}{P_1} 
\end{equation}
where $B/N=\alpha$. The eavesdropper is assumed to be unaware of the identity of the noise injecting Byzantines and considers each node $i$ to be Byzantine with a certain probability $\alpha$. Thus, the distribution of the data $\hat{y_i}$ at the eavesdropper under hypothesis $H_j$ can be approximated as IID multivariate Gaussian mixture with the same Gaussian parameters $\mathcal{N}(\tilde{\mu_j}^i,\tilde{\Sigma_j}^i)$ as above, however, with rescaled mixing probabilities $(\alpha P_1^j,\alpha P_2^j, 1-\alpha P_1^j-\alpha P_2^j)$. 
The eavesdropper makes the global decision about the phenomenon by considering the likelihood ratio test (LRT) which is given by

\begin{equation}
\label{l2}
\prod\limits_{i=1}^{N}\dfrac{f_1(\hat{y_i})}{f_0(\hat{y_i})}\quad \mathop{\stackrel{H_1}{\gtrless}}_{H_0} \quad  \dfrac{P_0}{P_1}. 
\end{equation}

Analyzing the performance of the likelihood ratio detector in~\eqref{l1} and~\eqref{l2} in a closed form is difficult in general. Thus, we use the modified deflection coefficient~\cite{MDC} in lieu of the probability of error of the system. Deflection coefficient reflects the output signal to noise ratio and widely used in optimizing the performance of detection systems. As stated earlier, the modified deflection coefficient is defined as

\begin{eqnarray*}
D(y_i)&=& (\mu_1^i-\mu_0^i)^T (\Sigma_1^i)^{-1} (\mu_1^i-\mu_0^i) 
\end{eqnarray*}
where $\mu_j^i$ and $\Sigma_j^i$ are the mean and the covariance matrix of $y_i$ under the hypothesis $H_j$, respectively. Using these notations, the modified deflection coefficient at the FC can be written as 

\begin{equation*}
D(FC)=B D(\tilde{y_i})+(N-B)D(y_i).
\end{equation*}
Dividing both sides of the above equation by $N$, we get

\begin{equation*}
D_{FC}=\alpha D(\tilde{y_i})+(1-\alpha)D(y_i)
\end{equation*}
where $D_{FC}=\frac{D(FC)}{N}$ and will be used as the performance metric as a surrogate for the probability of error.
Similarly, the modified deflection coefficient at the eavesdropper can be written as 

\begin{equation*}
D_{EV}=\frac{D(EV)}{N}=D(\hat{y_i}).
\end{equation*}

\section{System Design with Physical Layer Secrecy Guarantees}
\label{sec5}
Notice that, both $D_{FC}$ and $D_{EV}$ are functions of the compression ratio $c$ and artificial noise injection parameters $(\alpha,W_i)$ which are under the control of the FC. This motivates us to obtain the optimal values of system parameters
under a physical layer secrecy constraint. The problem can be formally stated as: 

\begin{equation}
\label{opt}
\begin{aligned}
& \underset{c,\alpha,W_i}{\text{maximize}}
& & \alpha D(\tilde{y_i})+(1-\alpha)D(y_i) \\
& \text{subject to}
& & D(\hat{y_i})\leq \tau \\
\end{aligned}
\end{equation}

\noindent where $c=M/P$ is the compression ratio. We refer to $D(\hat{y_i})\leq \tau$, where $\tau\geq 0$, as the physical layer secrecy constraint which reflects the security performance of the system. The case where $\tau=0$, or equivalently $D(\hat{y_i})= 0$, is referred to as the perfect secrecy constraint. In the wiretap channel literature, it is typical to consider the maximum degree of information achieved
by the main user (FC), while the information of the eavesdropper is exactly zero. This is commonly referred to as the perfect secrecy regime~\cite{wyner}. Next, we derive closed form expressions of the modified deflection coefficients at both the FC and the eavesdropper.

\begin{figure*}[t]
\centering
\subfigure[]{
\includegraphics[%
  width=0.35\textwidth,clip=true]{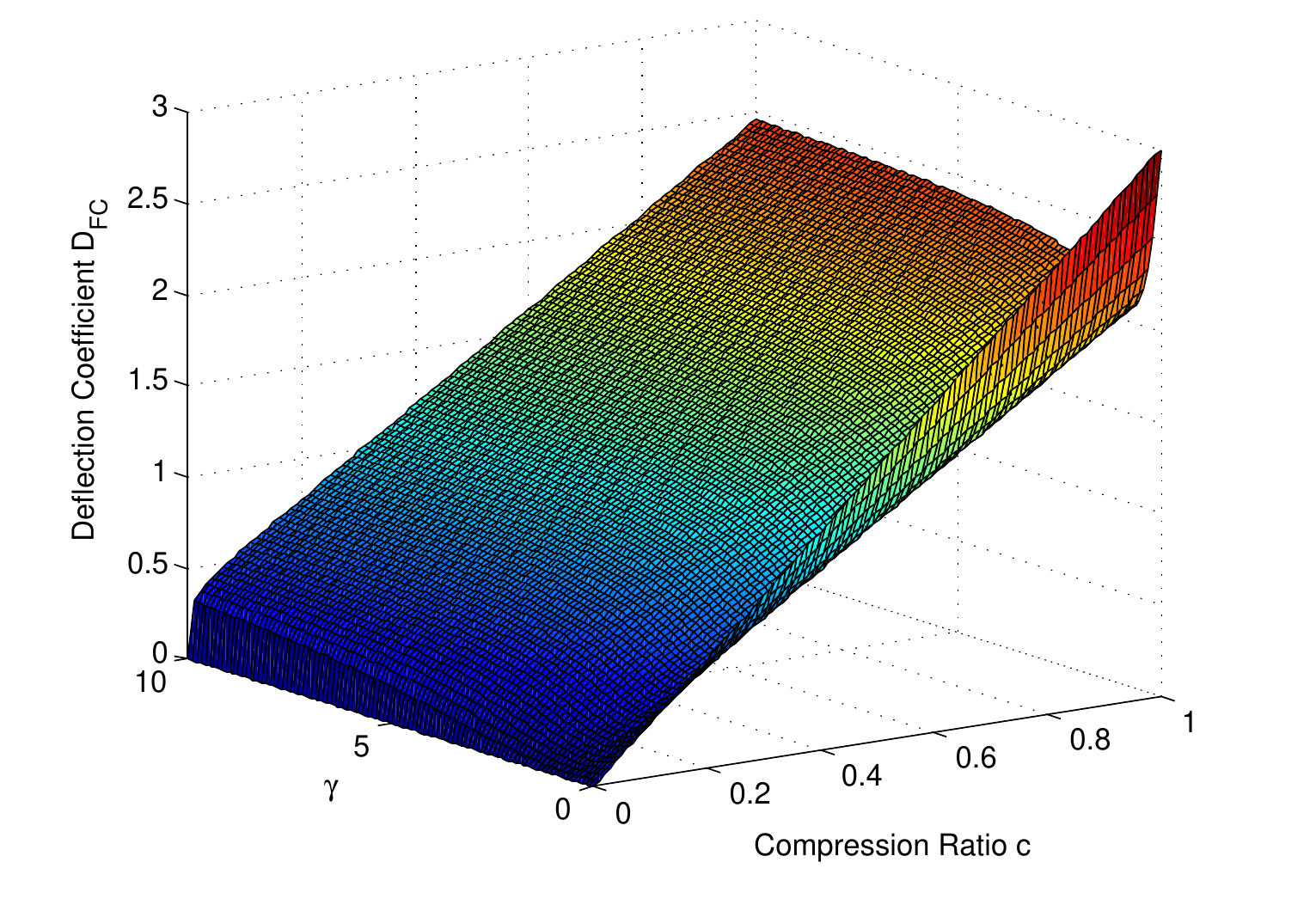}
\label{ruleVsFrac0}}
\hspace{0.1in}
\subfigure[] {
\includegraphics[%
  width=0.35\textwidth,clip=true]{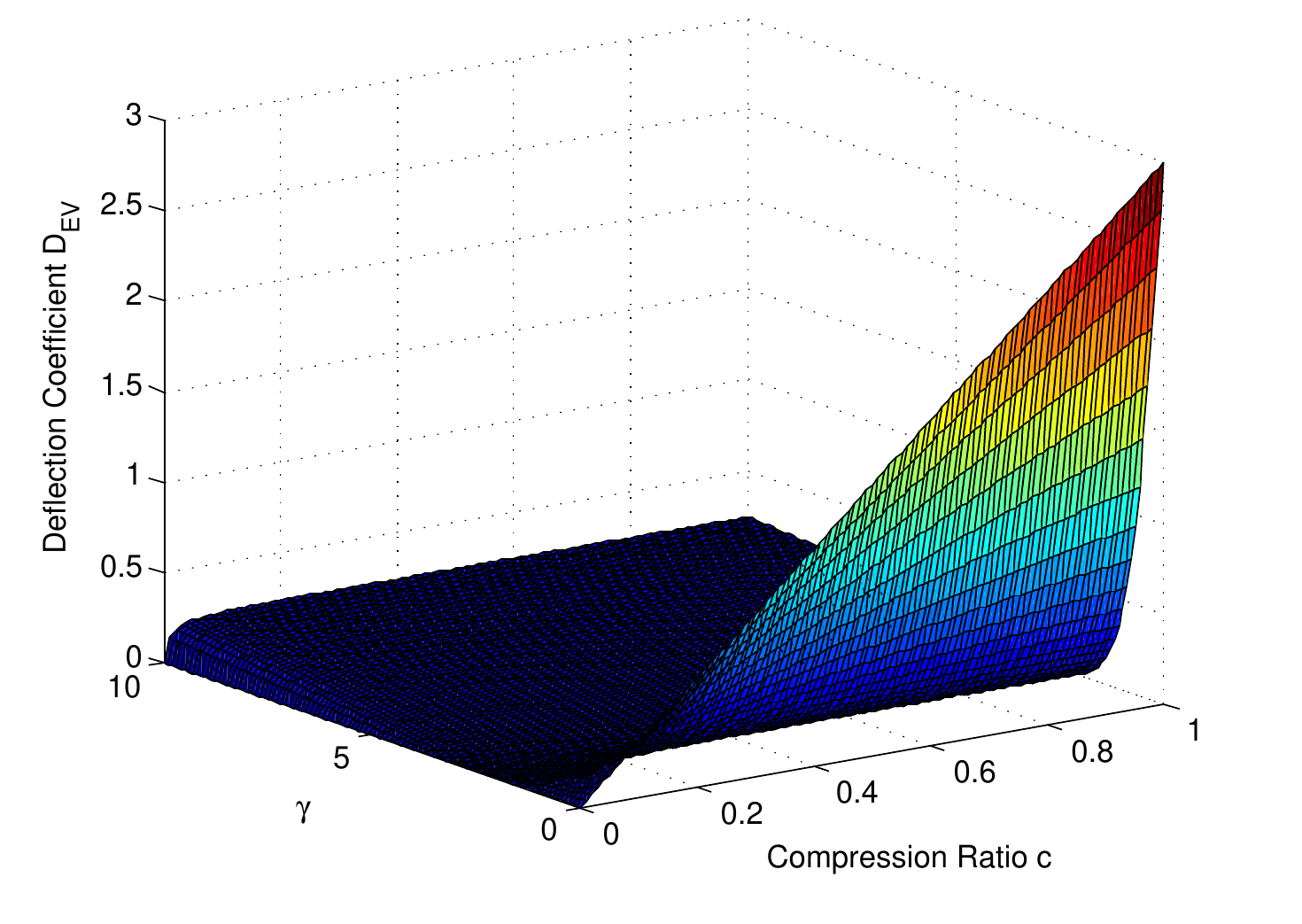}
\label{ruleVsFrac20} }
\caption{\vspace{-0.1in} Modified Deflection Coefficient analysis. \subref{ruleVsFrac0} $D_{FC}$ with varying $c$ and $\kappa$. \subref{ruleVsFrac20} $D_{EV}$ with varying $c$ and $\kappa$.}
\label{OptParamDesign0}
\vspace{-0.1in}
\end{figure*}
\vspace{-0.1in}
\begin{figure*}[t]
\centering
\subfigure[]{
\includegraphics[%
  width=0.35\textwidth,clip=true]{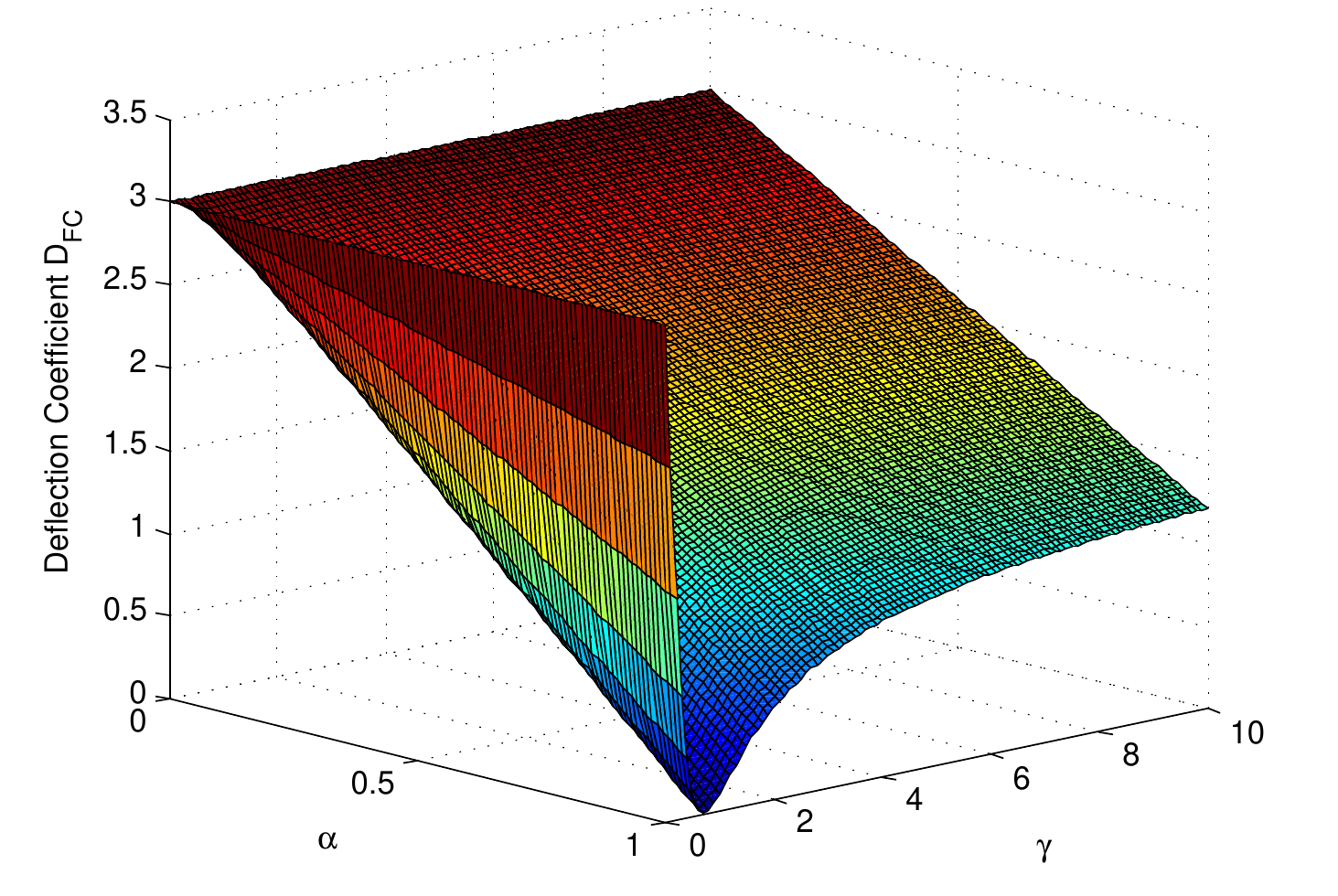}
\label{ruleVsFrac1}}
\hspace{0.1in}
\subfigure[] {
\includegraphics[%
  width=0.35\textwidth,clip=true]{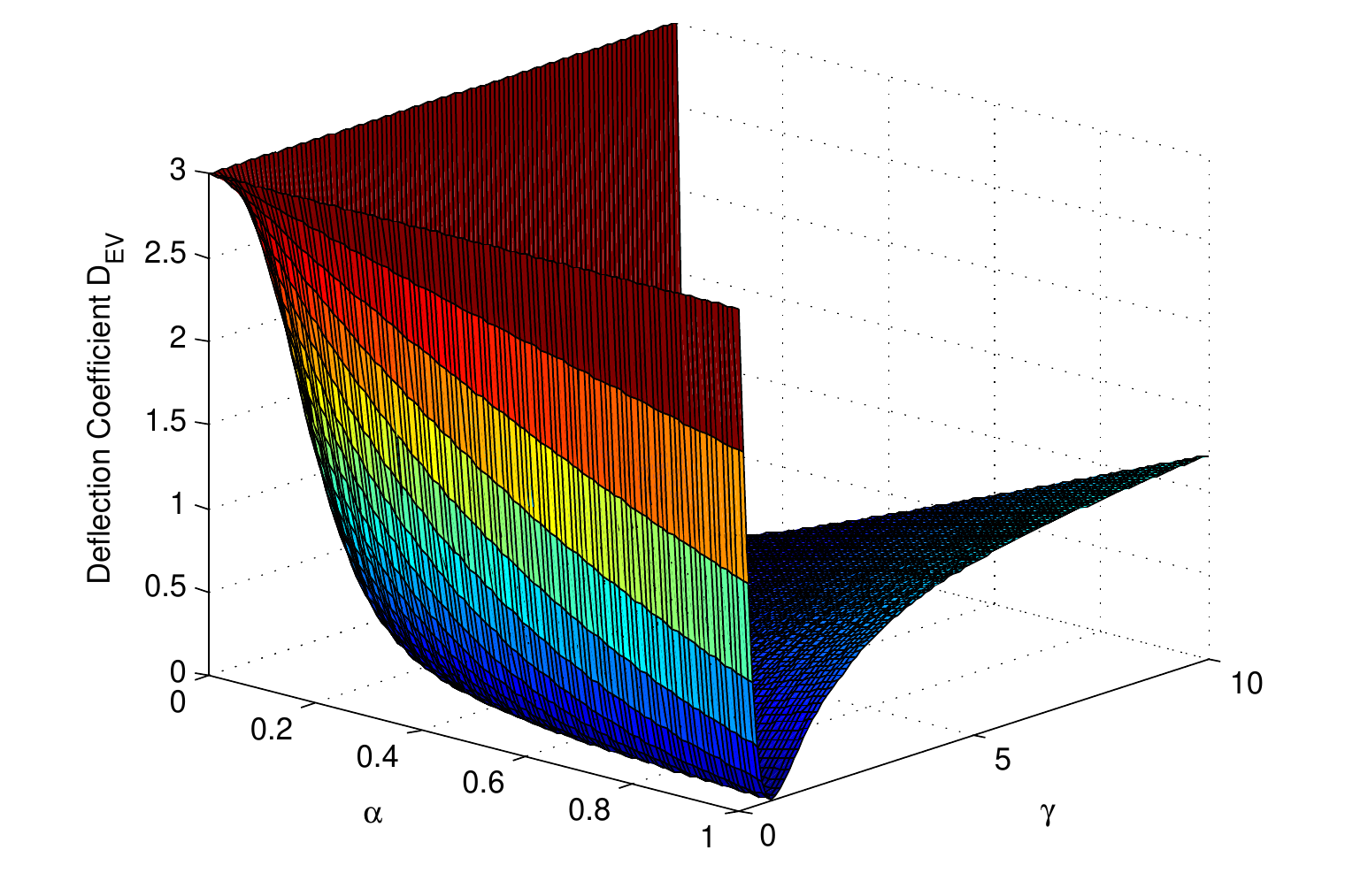}
\label{ruleVsFrac12} }
\caption{\vspace{-0.1in} Modified Deflection Coefficient analysis. \subref{ruleVsFrac1} $D_{FC}$ with varying $\alpha$ and $\kappa$. \subref{ruleVsFrac12} $D_{EV}$ with varying $\alpha$ and $\kappa$.}
\label{OptParamDesign1}
\vspace{-0.1in}
\end{figure*}

\subsection{Performance Analysis of Collaborative Compressive Detection with an Eavesdropper}

First, we derive closed form expressions of the modified deflection coefficients at the FC.

\subsubsection{A Closed Form Expression of the Modified Deflection Coefficient at the FC}
As stated earlier, the modified deflection coefficient at the FC is
\begin{equation*}
D_{FC}=\alpha D(\tilde{y_i})+(1-\alpha)D(y_i).
\end{equation*}

It can be shown that $$D(y_i)=\dfrac{\|\hat{P}\mu\|_2^2}{\gamma^{-1}+\beta^{-1}},$$ 
where $\hat{P}=\phi^T(\phi\phi^T)\phi$.

Next, to derive $D(\tilde{y_i})$, observe that $\tilde{y_i}$ is distributed as a multivariate Gaussian mixture with
\begin{eqnarray*}
\tilde{\mu}_0^i &=& (P_1^0-P_2^0)\phi D_i\\
\tilde{\mu}_1^i &=& (P_1^1-P_2^1)\phi D_i+\phi \mu\\
\tilde{\Sigma}_1^i &=& \sigma^2\phi\phi^T + \sum_{j=1}^{3} P_j^1(\tilde{\mu}_1^i(j)-\tilde{\mu}_1^i)(\tilde{\mu}_1^i(j)-\tilde{\mu}_1^i)^T 
\end{eqnarray*} 

where $\tilde{\mu}_1^i(1)=\phi(\mu+D_i)$, $\tilde{\mu}_1^i(2)=\phi(\mu-D_i)$, $\tilde{\mu}_1^i(3)=\phi\mu$, $P_3^1=1-P_1^1-P_2^1$ and $\sigma^2=(\alpha^{-1}+\beta^{-1}+\gamma^{-1})$. After some derivation, it can be shown that
\begin{equation*}
\tilde{\Sigma}_1^i = \sigma^2\phi\phi^T + P_t [\phi D_i D_i^T \phi^T]
\end{equation*} 
where $P_t=P_1^1+P_2^1-(P_1^1-P_2^1)^2$. Also, notice that $\tilde{\Sigma}_1^i$ is of the form $A+bb^T$. Now, using the Sherman-Morrison formula~\cite{sherman}, its inverse can be obtained to be
\begin{equation}
\label{cov}
(\tilde{\Sigma}_1^i)^{-1}=\frac{(\phi\phi^T)^{-1}}{\sigma^2}-\frac{P_t(\phi\phi^T)^{-1}\phi D_i D_i^T\phi^T(\phi\phi^T)^{-1}}{\sigma^4+\sigma^2P_tD_i^T\phi^T(\phi\phi^T)^{-1}\phi D_i}
\end{equation}
with $\sigma^2=(\alpha^{-1}+\beta^{-1}+\gamma^{-1})$.

Also,
\begin{equation}
\label{mu}
(\tilde{\mu}_1^i-\tilde{\mu}_0^i)=\phi \mu-P_b\phi D_i
\end{equation}
where $P_b=(P_1^0-P_2^0)+(P_2^1-P_1^1)$. Using~\eqref{cov},~\eqref{mu} and the fact that $D_i=\kappa \mu$ where $\kappa$ is referred to as the noise strength, the modified deflection coefficient $D(\tilde{y_i})$ can be derived to be\footnote{For $\mu=0$,  $D_{FC}\approx\alpha\frac{P_b^2}{P_t+\frac{\sigma^2}{\|\hat{P}d\|}}$ where $D_i=d,\forall i$.} 
\begin{equation}
D(\tilde{y_i})=(1-P_b \kappa)^2\frac{\|\hat{P} \mu\|_2^2}{\sigma^2}-P_t \gamma^2 (1-P_b \kappa)^2\frac{\|\hat{P} \mu\|_2^4}{\sigma^2 r_b}
\end{equation}

where $r_b=\sigma^2+ P_t \kappa^2 \|\hat{P} \mu\|_2^2$ and $\sigma^2=(\alpha^{-1}+\beta^{-1}+\gamma^{-1})$.

\begin{proposition}
\label{prp1}
Suppose that $\sqrt{\frac{P}{M}}\hat{P}$ provides an $\epsilon$-stable embedding of $(\mathcal{U},\{0\})$. Then the modified deflection coefficient at the FC for any $\mu\in \mathcal{U}$ can be approximated as 

\begin{equation}
D_{FC}\approx\alpha \frac{(1-P_b\kappa)^2}{\kappa^2 P_t+c^{-1}\frac{\sigma^2}{\| \mu\|_2^2}}+(1-\alpha)c \frac{\| \mu\|_2^2}{\sigma^2}
\end{equation}

\vspace{-0.05in}
where

$P_b=(P_1^0-P_2^0)+(P_2^1-P_1^1)$, $P_t=P_1^1+P_2^1-(P_1^1-P_2^1)^2$ and $\sigma^2=(\alpha^{-1}+\beta^{-1}+\gamma^{-1})$.
\end{proposition}
\begin{proof}
Using the fact that $\sqrt{\frac{P}{M}}\hat{P}$ provides an $\epsilon$-stable embedding of $(U,\{0\})$, for any $\mu\in\mathcal{U}$, $D(y_i)$ and $D(\tilde{y}_i)$ can be approximated as
\begin{eqnarray*}
D(y_i)&=& \frac{M}{P}\frac{\| \mu\|_2^2}{\sigma^2}\\
D(\tilde{y}_i)&=&  \frac{M}{P}\frac{\| \mu\|_2^2}{\sigma^2}(1-P_b\kappa)^2\left(1-\frac{M}{P}\frac{\| \mu\|_2^2}{r_b}\kappa^2P_t\right)
\end{eqnarray*}
where $r_b=\sigma^2+ P_t \kappa^2 \|\hat{P} \mu\|_2^2$ and $\sigma^2=(\alpha^{-1}+\beta^{-1}+\gamma^{-1})$.
Plugging in the above values in $D_{FC}=\alpha D(\tilde{y_i})+(1-\alpha)D(y_i)$ yields the desired result.
\end{proof}

Next, we derive the modified deflection coefficient at the eavesdropper.
\subsubsection{A Closed Form Expression of the Modified Deflection Coefficients at the Eavesdropper}
As stated earlier, the modified deflection coefficient of the eavesdropper is
\begin{equation*}
D_{EV}=D(\hat{y_i}).
\end{equation*}

Next, to derive $D(\hat{y_i})$, observe that $\hat{y_i}$ is distributed as a multivariate Gaussian mixture with
\begin{eqnarray*}
\hat{\mu}_0^i &=& \alpha(P_1^0-P_2^0)\phi D_i\\
\hat{\mu}_1^i &=& \alpha(P_1^1-P_2^1)\phi D_i+\phi \mu\\
\hat{\Sigma}_1^i &=& \sigma^2\phi\phi^T + \sum_{j=1}^{3} p_j^1(\hat{\mu}_1^i(j)-\hat{\mu}_1^i)(\hat{\mu}_1^i(j)-\hat{\mu}_1^i)^T 
\end{eqnarray*} 

with $\hat{\mu}_1^i(1)=\phi(\mu+D_i)$, $\hat{\mu}_1^i(2)=\phi(\mu-D_i)$, $\hat{\mu}_1^i(3)=\phi\mu$, $p_1^1=\alpha P_1^1$, $p_2^1=\alpha P_2^1$, $p_3^1=1-\alpha(P_1^1-P_2^1)$ and $\sigma^2=(\alpha^{-1}+\gamma^{-1}+\beta^{-1})$. Using these values, we state our result in the next proposition.\footnote{For $\mu=0$, $D_{EV}\approx\frac{(\alpha P_b)^2}{\alpha P_t^E+\frac{\sigma^2}{\|\hat{P}d\|}}$ where $D_i=d,\forall i$.}
\begin{proposition}
\label{prp2}
Suppose that $\sqrt{\frac{P}{M}}\hat{P}$ provides an $\epsilon$-stable embedding of $(\mathcal{U},\{0\})$. Then,  the modified deflection coefficient at the eavesdropper for any $\mu\in\mathcal{U}$ can be approximated as 

\begin{equation}
D_{EV}\approx\frac{(1-\alpha P_b\kappa)^2}{\alpha \kappa^2 P_t^E+c^{-1}\frac{\sigma^2}{\| \mu\|_2^2}}
\end{equation}

\vspace{-0.05in}
where

$P_b=(P_1^0-P_2^0)+(P_2^1-P_1^1)$, $P_t^E=P_1^1+P_2^1-\alpha(P_1^1-P_2^1)^2$ and $\sigma^2=(\alpha^{-1}+\beta^{-1}+\gamma^{-1})$.
\end{proposition}
\begin{proof}
The proof is similar to that of Proposition~\ref{prp1} and is, therefore, omitted.
\end{proof}

In general,  there is a trade-off between the detection performance and the security performance of the system.
To gain insights into this trade-off, in Figure~\ref{OptParamDesign0} we plot the modified deflection coefficient, both at the FC and at the eavesdropper, as a function of compression ratio $(c)$ and noise strength $(\kappa)$ when $\alpha=0.3,\; P_1^0=P_2^1=0.8,\; P_2^0=P_1^1=0.1$ and $\frac{\| \mu\|_2^2}{\sigma^2}=3$. Next, in Figure~\ref{OptParamDesign1} we plot the modified deflection coefficient, both at the FC and at the eavesdropper, as a function of the fraction of artificial noise injecting nodes $(\alpha)$ and noise strength $(\kappa)$ when $P_1^0=P_2^1=0.8,\; P_2^0=P_1^1=0.1$ and $\frac{M}{P}\frac{\| \mu\|_2^2}{\sigma^2}=3$.
It can be seen from Figure~\ref{OptParamDesign0} and Figure~\ref{OptParamDesign1} that $D_{FC}$ and $D_{EV}$ do not exhibit nice properties (monotonicity or convexity) with respect to the system parameters and, therefore, it is not an easy task to design the system parameters under an arbitrary physical layer secrecy constraint.  Also notice that, a specific case where the eavesdropper is completely blind deserves particular attention. This is referred to as the perfect secrecy regime, i.e., $D_{EV}=0$.
In the next subsection, we explore the problem of system design in a holistic manner in the perfect secrecy regime. More specifically, we are interested in analyzing the behavior of the modified deflection coefficient, both at the FC and at the eavesdropper, as a function of compression ratio $(c=M/P)$ and artificial noise injection parameters $(\alpha,W_i)$.  

\subsection{Optimal System Design Under Perfect Secrecy Constraint}
The goal of the designer is to maximize the detection performance $D_{FC}$, while ensuring perfect secrecy at the eavesdropper. The system design problem under perfect secrecy constraint can be formally stated as: 

\begin{equation}
\label{optper}
\begin{aligned}
& \underset{c,\alpha,W_i}{\text{maximize}}
& & \alpha D(\tilde{y_i})+(1-\alpha)D(y_i) \\
& \text{subject to}
& & D(\hat{y_i})=0 \\
\end{aligned}
\end{equation}
where $c$ is the compression ratio and $(\alpha,\;W_i)$ are the artificial noise injection parameters. This reduction of the search space, which arises as a natural consequence of the perfect secrecy constraint, has the additional benefit of simplifying the mathematical analysis. 
Next, we first explore the answer to the question: Does compression help in improving the security performance of the system? 

\subsubsection{Does Compression Help?}
We first consider the case where $\alpha P_b\kappa\neq1$. In this regime, for fixed values of $\alpha$, $P_b$ and $\kappa$, the modified deflection coefficient, both at the FC and the eavesdropper, is a monotonically increasing function of the compression ratio. In other words, $\frac{dD_{FC}}{dc}>0$ and $\frac{dD_{EV}}{dc}>0$. 
This suggests that compression improves the security performance at the expense of detection performance. More specifically, the FC would decrease the compression ratio until the physical layer secrecy constraint is satisfied. As a consequence, it will result in performance loss at the FC due to compression. In other words, there is a trade-off between the detection performance and the security performance of the  system. Observe that, $D_{EV}=0$ if and only if $\alpha P_b\kappa=1$ (ignoring the extreme conditions such as $c=0$ or $\kappa=\infty$) and, in this regime, $D_{FC}$ is a monotonically increasing function of the compression ratio $c$ and $D_{EV}$ is independent of the compression ratio $c$.
These results are summarized in the the following proposition.

\begin{proposition}
\label{prpc}
In the perfect secrecy regime $(i.e., \alpha P_b\kappa=1)$, $D_{FC}$ is a monotonically increasing function of the compression ratio $c$ and $D_{EV}$ is independent of the compression ratio $c$. When $ \alpha P_b\kappa\neq 1$, the modified deflection coefficient, both at the FC and the eavesdropper, is a monotonically increasing function of the compression ratio.
\end{proposition}
As mentioned above, in the perfect secrecy regime $D_{EV}$ is independent of the compression ratio $c$ and the network designer can fix $c=c_{max}$, where the value of $c_{max}$ may be dependent on the application of interest. The system design problem under perfect secrecy constraint can be reformulated as 
\begin{equation}
\label{opt2}
\underset{\alpha,W_i}{\arg \max} \; D_{FC}(c_{max},\kappa=1/(P_b\alpha)).
\end{equation}

Next, we analyze the behavior of the $D_{FC}(c_{max},\kappa=1/(P_b\alpha))$ as a function of artificial noise injection parameters $(\alpha,W_i)$.  

\subsubsection{Optimal Artificial Noise Injection Parameters}
\begin{figure}[!t]
\centering
    \includegraphics[width=3.5in]{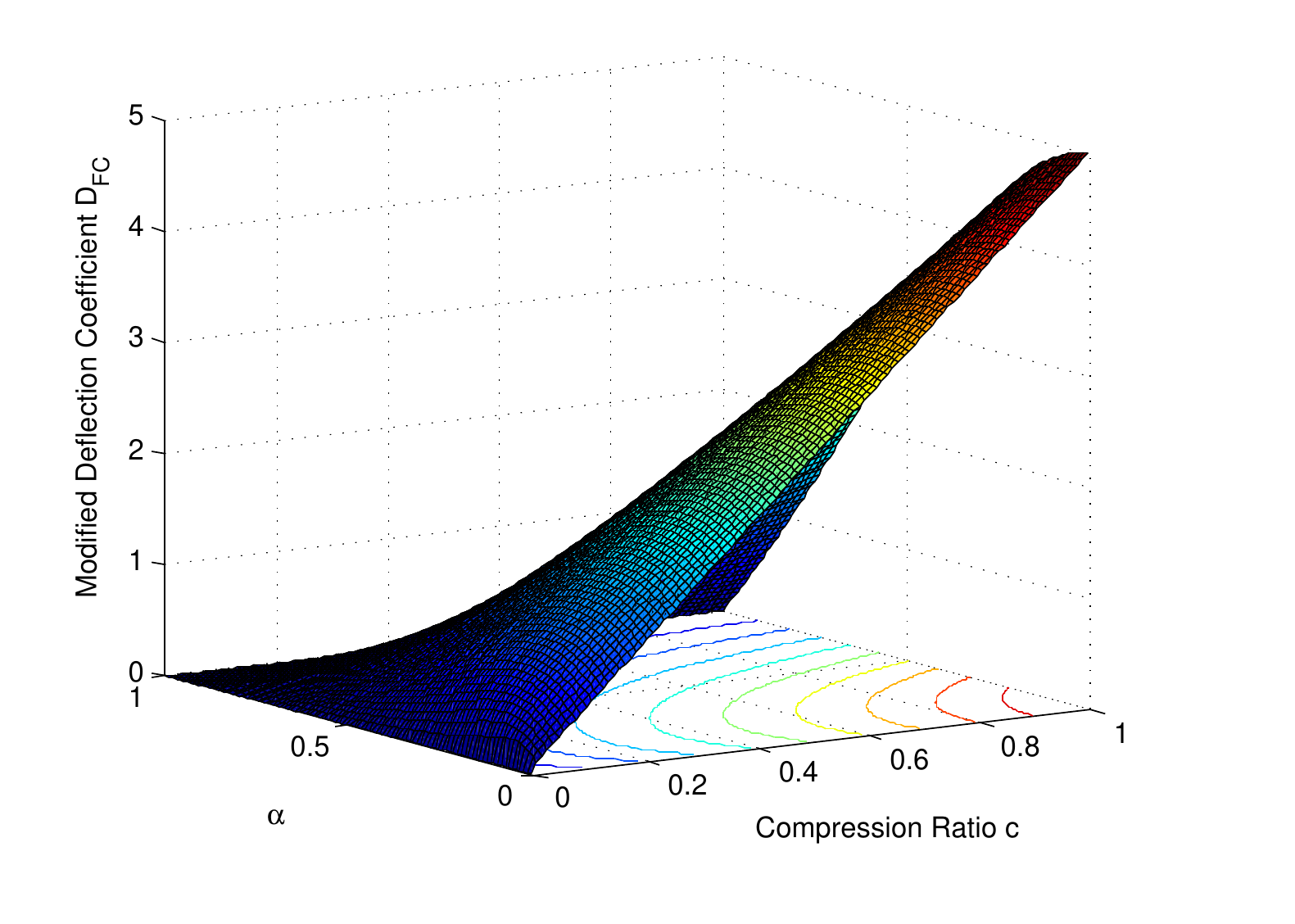}
\vspace{-0.2in}
    \caption{Modified Deflection Coefficient as a function of $\alpha$ and compression ratio $c=M/P$ for $SNR=5$dB in perfect secrecy regime.}
    \label{perfect}
\end{figure}
\begin{proposition}
\label{prpa}
In the high signal to noise ratio regime (defined as $\frac{\| \mu\|_2^2}{\sigma^2}>\frac{P_b^2}{P_t}$ where $\sigma^2=(\alpha^{-1}+\gamma^{-1}+\beta^{-1})$), the modified deflection coefficient at the FC, $D_{FC}$, is a monotonically decreasing function of the fraction of data falsifying nodes ($0<\alpha\leq 1$) under the perfect secrecy constraint.
\end{proposition}
\begin{proof}
Deflection coefficient at the FC under the perfect secrecy constraint can be expressed as
\begin{equation*}
D_{FC}(c_{max},\kappa=1/(P_b\alpha))=\frac{\alpha(1-\frac{1}{\alpha})^2}{\frac{P_t}{\alpha^2P_b^2}+\frac{1}{D}} + (1-\alpha) D
\end{equation*}
with $D=c_{max}\frac{\| \mu\|_2^2}{\sigma^2}$. Now, deriving the derivative of $D_{FC}$ with respect to $\alpha$ results in
\begin{small}
\begin{equation*}
\frac{dD_{FC}}{d\alpha}=\frac{P_tP_b^2D^2(1-4\alpha+\alpha^2)-DP_b^4\alpha^2-P_t^2 D^3}{(P_t D+\alpha^2P_b^2)^2}.
\end{equation*}
\end{small} 
\noindent Next, we show that $\frac{dD_{FC}}{d\alpha}<0$. 

First, let us define $F(\alpha)=x+\alpha^2\frac{1}{x}-(1+\alpha^2)$. It is easy to show that $F(\alpha)$ is a decreasing function of $\alpha$ if $x>1$. This also implies that $F(\alpha)>0$ if and only if $F(\alpha=1)>0 \Leftrightarrow (x+\frac{1}{x})>2$. Note that, $(x+\frac{1}{x})>2$, which follows from the fact that arithmetic mean is greater than the geometric mean. Having shown that $F(\alpha)>0$, we return back to showing that $\frac{dD_{FC}}{d\alpha}<0$. We start with the inequality

\begin{small}
\begin{eqnarray*}
&&
x+\alpha^2\frac{1}{x}-(1+\alpha^2)>0\\
&\Leftrightarrow&
x+\alpha^2\frac{1}{x}>(1+\alpha^2)\\
&\Rightarrow&
\frac{4\alpha}{(1+\alpha^2)}+\frac{x}{(1+\alpha^2)}+\frac{\alpha^2}{(1+\alpha^2)}\frac{1}{x}>1\\
\end{eqnarray*}
\end{small}

Now, if we plug in $x=D\frac{P_t}{P_{B}^2}$ in the above inequality and rearrange the terms we get

\begin{small} 
\begin{equation*}
\frac{P_tP_b^2D^2(1-4\alpha+\alpha^2)-DP_b^4\alpha^2-P_t^2 D^3}{(P_t D+\alpha^2P_b^2)^2}<0
\end{equation*}
\end{small} 

which is true if $x=D\frac{P_t}{P_{B}^2}>1$.
\end{proof}

Notice that, Proposition~\ref{prpc} suggests that to maximize the modified deflection coefficient $D_{FC}$ under the perfect secrecy constraint~\eqref{optper}, the network designer should choose the value of $\alpha$ as low as possible under the constraint that $\alpha>0$ and accordingly increase $\kappa$ to satisfy $\alpha P_b\kappa=1$. In practice, $\alpha_{min}$ may be dependent on the application of interest.

Next, to gain insights into Proposition~\ref{prpc} and \ref{prpa}, we present illustrative examples that corroborate our results. In Figure~\ref{perfect}, we plot $D_{FC}$ as a function of fraction of noise injection nodes $\alpha$ and compression ratio $c$ in the perfect secrecy regime when $P_1^0=P_2^1=0.8,\; P_2^0=P_1^1=0.1$ and $\frac{\| \mu\|_2^2}{\sigma^2}=5dB$. It can be seen from the figure that $D_{FC}$ is a monotonically increasing function of $c$ and a monotonically decreasing function of $\alpha$. 

Next, we analyze the behavior of $D_{FC}$ as a function of artificial noise variance $\gamma^{-1}$. This analysis will help us in determining the optimal artificial noise injection parameters.

\begin{proposition}
\label{prps}
In the high signal to noise ratio regime (defined as $\frac{\| \mu\|_2^2}{\sigma^2}>\frac{P_b^2}{P_t}$ where $\sigma^2=(\alpha^{-1}+\gamma^{-1}+\beta^{-1})$)
with perfect secrecy constraint $(i.e.,\;\alpha P_b\kappa=1)$, the optimal artificial noise is a deterministic signal with value $\frac{\mu}{\alpha_{min} P_b}$, i.e., $f_{W_i}(w_i)=\delta(w_i-\frac{\mu}{\alpha_{min} P_b})$.
\end{proposition}
\begin{proof}
The proof follows from Proposition~\ref{prpa} and the fact that $D_{FC}$ is a monotonically decreasing function of the variance $\gamma^{-1}$ of the artificial noise, i.e., $\frac{d D_{FC}}{d\gamma^{-1}}<0$. 
\end{proof}
\begin{figure}[t]
\centering
    \includegraphics[width=3.5in]{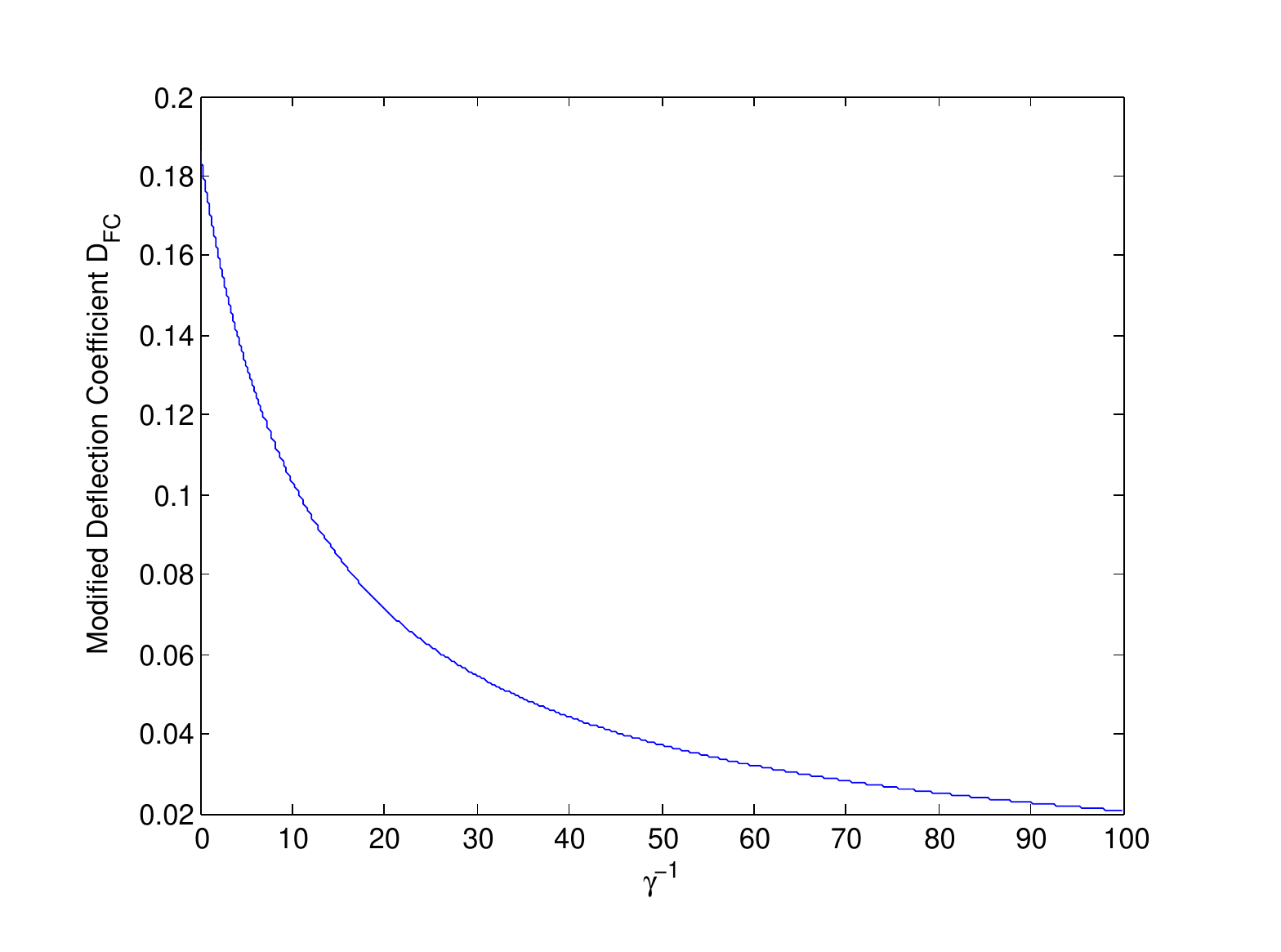}
\vspace{-0.2in}
    \caption{Modified Deflection Coefficient as a function of $\gamma^{-1}$ in perfect secrecy regime.}
    \label{perfectsigma}
\end{figure}

In Figure~\ref{perfectsigma}, we plot the modified deflection coefficient at the FC as a function of the variance of the artificial noise when $P_1^0=P_2^1=0.8,\; P_2^0=P_1^1=0.1$ and $(c,\alpha)=(0.2,0.3)$. We assume that the signal of interest is $s\sim\mathcal{N}(\mu,I_P)$ with $\|\mu\|_2^2=5$ and noise $v_i\sim\mathcal{N}(0,10 I_P)$.
It can be seen that $D_{FC}$ is a monotonically decreasing function of the artificial noise variance $\gamma^{-1}$. This
observation implies that the optimal artificial noise is a deterministic signal.
Using these results, the solution of the optimization problem~\eqref{optper} is summarized in the following theorem.
\begin{theorem}
\label{main}
To maximize the modified deflection coefficient at the FC under the perfect secrecy constraint, the network designer should choose $c=c_{max}$, $\alpha=\alpha_{min}$ and deterministic artificial noise with value $\frac{\mu}{\alpha_{min} P_b}$.
\end{theorem} 

Notice that, Theorem~\ref{main} suggests that to maximize the modified deflection coefficient $D_{FC}$ under the perfect secrecy constraint~\eqref{optper}, the network designer should choose the value of $\alpha$ as low as possible under the constraint that $\alpha>0$ and accordingly increase $\kappa$ to satisfy $\alpha P_b\kappa=1$. Also, the optimal artificial noise is a deterministic signal with value $\frac{\mu}{\alpha_{min} P_b}$, i.e., $f_{W_i}(w_i)=\delta(w_i-\frac{\mu}{\alpha_{min} P_b})$. 

\section{Conclusion and Future Work}
\label{sec6}
We considered the problem of collaborative compressive detection under a physical layer secrecy constraint.
First, we proposed the collaborative compressive detection framework and showed that through collaboration the loss due to compression when using a single node can be recovered. Second, we studied the problem where the network works in the presence of an eavesdropper. We proposed the use of artificial noise injection techniques to improve security performance. We also considered the problem of
determining optimal system parameters which maximize the
detection performance at the FC, while ensuring perfect secrecy
at the eavesdropper. Optimal system parameters with perfect secrecy guarantees were obtained in a closed form. 
There are still many interesting questions that remain to be explored in the future work such as an analysis of the problem in scenarios where the perfect secrecy constraint is relaxed. Note that, some analytical methodologies
used in this paper are certainly exploitable for studying more general detection problems such as detection of non Gaussian signals in correlated noise.
Other questions such as the case where communication channels are noisy can also be investigated.
\bibliographystyle{IEEEtran}
\bibliography{Conf,Book,Journal}
\end{document}